\documentclass[letterpaper,onecolumn,11pt,accepted=2021-08-06]{quantumarticle}
\pdfoutput=1
\usepackage[english]{babel}

\usepackage[T1]{fontenc}
\usepackage{amssymb,amsmath,amsthm,amsfonts}
\usepackage{mathrsfs}
\usepackage{mathtools}
\usepackage{enumitem}
\usepackage[numbers,comma,sort&compress]{natbib}
\usepackage{caption}
\usepackage{float}
\usepackage[ruled,vlined,linesnumbered]{algorithm2e}
\usepackage{physics}
\usepackage{footnote}
\usepackage{xcolor}
\usepackage{thm-restate}
\usepackage[colorlinks]{hyperref}

\usepackage{tikz}
\usepackage{lipsum}

\newcommand{\eqn}[1]{(\ref{eqn:#1})}

\newcommand{\rem}[1]{\hyperref[rem:#1]{Remark~\ref*{rem:#1}}}
\newcommand{\thm}[1]{\hyperref[thm:#1]{Theorem~\ref*{thm:#1}}}
\newcommand{\cor}[1]{\hyperref[cor:#1]{Corollary~\ref*{cor:#1}}}
\newcommand{\defn}[1]{\hyperref[defn:#1]{Definition~\ref*{defn:#1}}}
\newcommand{\lem}[1]{\hyperref[lem:#1]{Lemma~\ref*{lem:#1}}}
\newcommand{\prop}[1]{\hyperref[prop:#1]{Proposition~\ref*{prop:#1}}}
\newcommand{\fig}[1]{\hyperref[fig:#1]{Figure~\ref*{fig:#1}}}
\newcommand{\tab}[1]{\hyperref[tab:#1]{Table~\ref*{tab:#1}}}
\newcommand{\algo}[1]{\hyperref[algo:#1]{Algorithm~\ref*{algo:#1}}}
\renewcommand{\sec}[1]{\hyperref[sec:#1]{Section~\ref*{sec:#1}}}
\newcommand{\append}[1]{\hyperref[append:#1]{Appendix~\ref*{append:#1}}}
\newcommand{\fac}[1]{\hyperref[fac:#1]{Fact~\ref*{fac:#1}}}
\newcommand{\lin}[1]{\hyperref[lin:#1]{Line~\ref*{lin:#1}}}
\newcommand{\fnote}[1]{\hyperref[fnote:#1]{Footnote~\ref*{fnote:#1}}}

\SetKwInput{KwInput}{Input}
\SetKwInput{KwOutput}{Output}

\newcommand{\vect}[1]{\ensuremath{\mathbf{#1}}}
\newcommand{\x}{\ensuremath{\mathbf{x}}}

\newcommand{\N}{\mathbb{N}}
\newcommand{\R}{\mathbb{R}}
\renewcommand{\P}{\mathbb{P}}

\DeclareMathOperator{\poly}{poly}
\DeclareMathOperator{\diag}{diag}

\newcommand{\tnb}{\tilde{\nabla}}

\renewcommand{\d}{\mathrm{d}}
\def\Tr{\operatorname{Tr}}

\newtheorem{theorem}{Theorem}

\newtheorem{lemma}{Lemma}
\newtheorem{proposition}{Proposition}

\newtheorem{corollary}{Corollary}
\newtheorem{remark}{Remark}

\begin{document}

\title{Quantum Algorithms for Escaping from Saddle Points}

\author{Chenyi Zhang$^{*}$}
\affiliation{Institute for Interdisciplinary Information Sciences, Tsinghua University, Beijing, China}
\author{Jiaqi Leng$^{*}$}
\affiliation{Department of Mathematics and Joint Center for Quantum Information and Computer Science, University of Maryland, College Park, MD, USA}
\author{Tongyang Li$^{\dagger}$}
\affiliation{Center on Frontiers of Computing Studies, Peking University, Beijing, China}
\affiliation{Center for Theoretical Physics, Massachusetts Institute of Technology, Cambridge, MA, USA}
\affiliation{Department of Computer Science and Joint Center for Quantum Information and Computer Science, University of Maryland, College Park, MD, USA}
\maketitle

\begin{abstract}
\def\thefootnote{*}\footnotetext{Equal contribution.}
\def\thefootnote{$\dagger$}\footnotetext{Corresponding author. Email: tongyangli@pku.edu.cn}
We initiate the study of quantum algorithms for escaping from saddle points with provable guarantee. Given a function $f\colon\R^{n}\to\R$, our quantum algorithm outputs an $\epsilon$-approximate second-order stationary point using $\tilde{O}(\log^{2} (n)/\epsilon^{1.75})$\def\thefootnote{1}\footnote{The $\tilde{O}$ notation omits poly-logarithmic terms, i.e., $\tilde{O}(g)=O(g\poly(\log g))$.} queries to the quantum evaluation oracle (i.e., the zeroth-order oracle). Compared to the classical state-of-the-art algorithm by Jin et al. with $\tilde{O}(\log^{6} (n)/\epsilon^{1.75})$ queries to the gradient oracle (i.e., the first-order oracle), our quantum algorithm is polynomially better in terms of $\log n$ and matches its complexity in terms of $1/\epsilon$. Technically, our main contribution is the idea of replacing the classical perturbations in gradient descent methods by simulating quantum wave equations, which constitutes the improvement in the quantum query complexity with $\log n$ factors for escaping from saddle points. We also show how to use a quantum gradient computation algorithm due to Jordan to replace the classical gradient queries by quantum evaluation queries with the same complexity. Finally, we also perform numerical experiments that support our theoretical findings.
\end{abstract}


\section{Introduction}
Nonconvex optimization is a central research topic in optimization theory, mainly because the loss functions in many machine learning models (including neural networks) are typically nonconvex. However, finding a global optimum of a nonconvex function is NP-hard in general. Instead, many theoretical works focus on finding local optima, since there are landscape results suggesting that local optima are nearly as good as the global optima for many learning problems~\cite{ge2015escaping,ge2017optimization,ge2018learning,bhojanapalli2016global,ge2016matrix,hardt2018gradient}. On the other hand, it is known that saddle points (and local maxima) can give highly suboptimal solutions in many problems~\cite{jain2017global,sun2018geometric}. Furthermore, saddle points are ubiquitous in high-dimensional nonconvex optimization problems~\cite{dauphin2014identifying,bray2007statistics,fyodorov2007replica}.

Therefore, one of the most important problems in nonconvex optimization is to \emph{escape from saddle points}. Suppose we have a twice-differentiable function $f\colon\R^{n}\to\R$ such that
\begin{itemize}
\item $f$ is $\ell$-smooth: $\|\nabla f(\x_{1})-\nabla f(\x_{2})\|\leq\ell\|\x_{1}-\x_{2}\|\quad\forall\,\x_{1},\x_{2}\in\R^{n}$,
\item $f$ is $\rho$-Hessian Lipschitz: $\|\nabla^{2} f(\x_{1})-\nabla^{2} f(\x_{2})\|\leq\rho\|\x_{1}-\x_{2}\|\quad\forall\,\x_{1},\x_{2}\in\R^{n}$;
\end{itemize}
the goal is to find an $\epsilon$-approximate local minimum $\x_{\epsilon}$ (also known as an $\epsilon$-approximate second-order stationary point) such that\footnote{In general, we can ask for an $(\epsilon_{1},\epsilon_{2})$-approximate local minimum $\x$ such that $\|\nabla f(\x)\|\leq\epsilon_{1}$ and $\lambda_{\min}(\nabla^{2}f(\x))\geq-\epsilon_{2}$. The scaling in \eqn{eps-approx-local-min} was first adopted by~\cite{nesterov2006cubic} and is taken as a standard by subsequent works~\cite{jin2017escape,jin2018accelerated,jin2019stochastic,xu2017neon,xu2018first,carmon2018accelerated,agarwal2017finding,tripuraneni2018stochastic,fang2019sharp}.}
\begin{align}\label{eqn:eps-approx-local-min}
\|\nabla f(\x_{\epsilon})\|\leq\epsilon,\quad\lambda_{\min}(\nabla^{2}f(\x_{\epsilon}))\geq-\sqrt{\rho\epsilon}.
\end{align}
Intuitively, this means that at $\x_{\epsilon}$, the gradient is small with norm being at most $\epsilon$, and the Hessian is close to be positive semi-definite with the smallest eigenvalue being at least $-\sqrt{\rho\epsilon}$.

There have been two main focuses on designing algorithms for escaping from saddle points. First, algorithms with good performance in practice are typically dimension-free or almost dimension-free (i.e., having $\poly(\log n)$ dependence), especially considering that most machine learning models in the real world have enormous dimensions. Second, practical algorithms prefer simple oracle access to the nonconvex function. If we are given a Hessian oracle of $f$, which takes $\x$ as the input and outputs $\nabla^{2} f(\x)$, we can find an $\epsilon$-approximate local minimum by second-order methods; for instance, Ref.~\cite{nesterov2006cubic} took $O(1/\epsilon^{1.5})$ queries. However, because the Hessian is an $n\times n$ matrix, its construction takes $\Omega(n^{2})$ cost in general. Therefore, it has become a notable interest to escape from saddle points using simpler oracles.

A seminal work along this line was by Ge et al.~\cite{ge2015escaping}, which can find an $\epsilon$-approximate local minimum satisfying \eqn{eps-approx-local-min} only using the first-order oracle, i.e., gradients. Although this paper has a $\poly(n)$ dependence in the query complexity of the oracle, the follow-up work by~\cite{jin2017escape} achieved to be almost dimension-free with complexity $\tilde{O}(\log^{4} (n)/\epsilon^{2})$, and the state-of-the-art result takes $\tilde{O}(\log^{6} (n)/\epsilon^{1.75})$ queries~\cite{jin2018accelerated}. However, these results suffer from a significant overhead in terms of $\log n$, and it has been an open question to keep both the merits of using only the first-order oracle as well as being close to dimension-free~\cite{jordan2017talk}.

On the other hand, quantum computing is a rapidly advancing technology. In particular, the capability of quantum computers is dramatically increasing and recently reached ``quantum supremacy''~\cite{Preskill2018NISQ} by Google~\cite{arute2019supremacy}. However, at the moment the noise of quantum gates prevents current quantum computers from being directly useful in practice; consequently, it is also of significant interest to understand quantum algorithms from a theoretical perspective for paving its way to future applications.

In this paper, we explore quantum algorithms for escaping from saddle points. This is a mutual generalization of both classical and quantum algorithms for optimization:
\begin{itemize}
\item For classical optimization theory, since many classical optimization methods are physics-motivated, including Nesterov's momentum-based methods~\cite{nesterov1983method}, Hamiltonian Monte Carlo~\cite{gao2018global} or stochastic gradient Langevin dynamics~\cite{zhang2017hitting}, etc., the elevation from classical mechanics to quantum mechanics can potentially bring more observations on designing fast \emph{quantum-inspired classical algorithms}. In fact, quantum-inspired classical machine learning algorithms have been an emerging topic in theoretical computer science~\cite{tang2018quantum2,tang2019quantum,chia2020SDP,chia2019framework,chia2020issac,gilyen2020improved,shao2021faster}, and it is worthwhile to explore relevant classical algorithms for optimization.

\item For quantum computing, the vast majority of previous quantum optimization algorithms had been devoted to convex optimization with the focuses on semidefinite programs~\cite{brandao2016quantum,vanApeldoorn2017quantum,brandao2017SDP,vanApeldoorn2018SDP,kerenidis2020quantum} and general convex optimization~\cite{vanApeldoorn2020optimization,chakrabarti2020optimization}; these results have at least a $\sqrt{n}$ dependence in their complexities, and their quantum algorithms are far from dimension-free methods. Up to now, little is known about quantum algorithms for nonconvex optimization.

However, there are inspirations that quantum speedups in nonconvex scenarios can potentially be more significant than convex scenarios. In particular, \emph{quantum tunneling} is a phenomenon in quantum mechanics where the wave function of a quantum particle can tunnel through a potential barrier and appear on the other side with significant probability. This very much resembles escaping from poor landscapes in nonconvex optimization. Moreover, quantum algorithms motivated by quantum tunneling will be essentially different from those motivated by the Grover search~\cite{grover1996fast}, and will demonstrate significant novelty if the quantum speedup compared to the classical counterparts is more than quadratic.
\end{itemize}

\subsection{Contributions}
Our main contribution is a quantum algorithm that can find an $\epsilon$-approximate local minimum of a function $f\colon\R^{n}\to\R$ that is smooth and Hessian Lipschitz. Compared to the classical state-of-the-art algorithm by~\cite{jin2018accelerated} using $\tilde{O}(\log^{6}(n)/\epsilon^{1.75})$ queries to the gradient oracle (i.e., the first-order oracle), our quantum algorithm achieves an improvement in query complexity with $\log n$ factors. Furthermore, our quantum algorithm only takes queries to the \emph{quantum evaluation oracle} (i.e., the zeroth-order oracle), which is defined as a unitary map $U_{f}$ on $\R^{n}\otimes\R$ such that for any $\ket{\x}\in\R^{n}$,
\begin{align}\label{eqn:quantum-evaluation}
U_{f}(\ket{\x}\otimes \ket{0})=\ket{\x}\otimes \ket{f(\x)}.
\end{align}
Furthermore, for any $m\in\mathbb{N}$, $\ket{\x_{1}},\ldots,\ket{\x_{m}}\in\mathbb{R}^{n}$, and $\vect{c}\in\mathbb{C}^{m}$ such that $\sum_{i=1}^{m}|\vect{c}_{i}|^{2}=1$,
\begin{align}
U_{f}\Big(\sum_{i=1}^{m}\vect{c}_{i}\ket{\x_{i}}\otimes \ket{0}\Big)=\sum_{i=1}^{m}\vect{c}_{i}\ket{\x_{i}}\otimes \ket{f(\x_{i})}.
\end{align}
If we measure this quantum state, we get $f(\vect{x}_{i})$ with probability $|\vect{c}_{i}|^{2}$. Compared to the classical evaluation oracle (i.e., $m=1$), the quantum evaluation oracle allows the ability to query different locations in \emph{superposition}, which is the essence of speedups from quantum algorithms. In addition, if the classical evaluation oracle can be implemented by explicit arithmetic circuits, the quantum evaluation oracle in \eqn{quantum-evaluation} can be implemented by quantum arithmetic circuits of about the same size. As a result, it is the standard assumption in previous literature on quantum algorithms for various optimization problems, including quadratic forms~\cite{jordan2008quantum}, basin hopper~\cite{bulger2005quantum}, and general convex optimization~\cite{vanApeldoorn2020optimization,chakrabarti2020optimization}. Subsequently, we adopt it here for general nonconvex optimization.

\begin{theorem}[Main result, informal]\label{thm:main-intro}
Our quantum algorithm finds an $\epsilon$-approximate local minimum using $\tilde{O}(\log^{2} (n)/\epsilon^{1.75})$ queries to the quantum evaluation oracle~\eqn{quantum-evaluation}.
\end{theorem}

Technically, our work is inspired by both the perturbed gradient descent (PGD) algorithm in~\cite{jin2017escape,jin2019stochastic} and the perturbed accelerated gradient descent (PAGD) algorithm in~\cite{jin2018accelerated}. To be more specific, PGD applies gradient descent iteratively until it reaches a point with small gradient. It can potentially be a saddle point, so PGD applies uniform perturbation in a small ball centered at that point and then continues the GD iterations. It can be shown that with an appropriate choice of the radius, PGD can shake the point off from the saddle and converge to a local minimum with high probability. The PAGD in~\cite{jin2018accelerated} adopts the similar perturbation idea, but the GD is replaced by Nesterov's AGD~\cite{nesterov1983method}.

Our quantum algorithm is built upon PGD and PAGD and shares their simplicity of being single-loop, but we propose two main modifications. On the one hand, for the perturbation steps for escaping from saddle points, we replace the uniform perturbation by evolving a quantum wave function governed by the Schr\"{o}dinger equation and using the measurement outcome as the perturbed result. Intuitively, the Schr\"odinger equation screens the local geometry of a saddle point through wave interference, which results in a phenomenon that the wave packet disperses rapidly along the directions with significant function value decrease. Specifically, quantum mechanics finds the negative curvature directions more efficiently than the classical counterpart: for a constant $\epsilon$, the classical PGD and PAGD take $O(\log n)$ steps to decrease the function value by $\Omega(1/\log^3 n)$ and $\Omega(1/\log^5 n)$ with high probability, respectively. Quantumly, the simulation of the Schr\"{o}dinger equation for time $t$ takes $\tilde{O}(t\log n)$ evaluation queries,\footnote{In general, the query complexity of quantum simulation depends on the properties of the Hamiltonian, i.e., norm, sparsity, etc. In our case, the Hamiltonian takes the form $H = A + B$, where $A$ is of norm $\alpha_A = \poly(n)$ but is independent of $f$, and $B$ is a diagonal matrix (so its sparsity is $1$) that encodes the evaluations of $f$. It turns out that the interaction picture simulation technique~\cite{low2018hamiltonian} is particularly suitable for this circumstance, and we only need $\tilde{O}(t\log n)$ queries to $f$. For details, see \sec{simulation}.} but simulation for time $t=O(\log n)$ suffices to decrease the function value by $\Omega(1)$ with high probability. See \prop{QS-effectiveness} and \thm{QuantumSimulationAGD}.

In addition, we replace the gradient descent steps by a quantum algorithm for computing gradients using also quantum evaluation queries. The idea was initiated by Jordan in Ref.~\cite{jordan2005fast} which computed the gradient at a point by applying the quantum Fourier transform on a mesh near the point. Prior work has applied Jordan's algorithm to general convex optimization~\cite{chakrabarti2020optimization,vanApeldoorn2020optimization}; we follow the same path by conducting a detailed analysis (see \thm{ESCGDJordan}) showing how we replace classical gradient queries by the same number of quantum evaluation queries in nonconvex optimization.

It is worth highlighting that our quantum algorithm enjoys the following two nice features:
\begin{itemize}
\item \emph{Classical-quantum hybrid:} In \algo{PAGDQS} and \algo{PGD-Jordan}, the transition between consecutive iterations is still classical, while the only quantum computing part happens inside each iteration for replacing the classical uniform perturbation. Such feature is friendly for the implementation on near-term quantum computers.

\item \emph{Robustness:} Our quantum algorithm is robust from two aspects. On the one hand, we can even escape from an approximate saddle point by evolving the Schr\"{o}dinger equation (see \prop{QS-effectiveness}). On the other hand, \thm{ESCGDJordan} essentially shows the robustness of escaping from saddle points by even noisy gradient descents, which may be of independent interest.
\end{itemize}

Finally, we perform numerical experiments that support our theoretical findings. Specifically, we observe the dispersion of quantum wave packets along the negative curvature direction in various landscapes. In a comparative study, our PGD with quantum simulation outperforms the classical PGD with a higher probability of escaping from saddle points and fewer iteration steps. We also compare the dimension dependence of classical and quantum algorithms in a model question with dimensions varying from 10 to 1000, and our quantum algorithm achieves a better dimension scaling overall.

\begin{table}[htbp]
\centering
\resizebox{0.6\columnwidth}{!}{
\begin{tabular}{ccc}
\hline
Reference & Queries & Oracle \\ \hline
\cite{nesterov2006cubic,curtis2017trust} & $O(1/\epsilon^{1.5})$ & Hessian \\ \hline
\cite{agarwal2017finding,carmon2018accelerated} & $\tilde{O}(\log (n)/\epsilon^{1.75})$ & Hessian-vector product \\ \hline
\cite{jin2017escape,jin2019stochastic} & $\tilde{O}(\log^{4}(n)/\epsilon^{2})$ & Gradient \\ \hline
\cite{jin2018accelerated} & $\tilde{O}(\log^{6}(n)/\epsilon^{1.75})$ & Gradient \\ \hline\hline
\textbf{this work} & $\tilde{O}(\log^{2} (n)/\epsilon^{1.75})$ & Quantum evaluation \\
\hline
\end{tabular}
}
\caption{A summary of the state-of-the-art results on finding approximate second-order stationary points. The query complexities are highlighted in terms of the dimension $n$ and the precision $\epsilon$.}
\label{tab:main}
\end{table}

\subsection{Related Work}
Escaping from saddle points by gradients was initiated by~\cite{ge2015escaping} with complexity $O(\poly(n/\epsilon))$. The follow-up work by~\cite{levy2016power} improved it to $O(n^{3}\poly(1/\epsilon))$, but it is still polynomial in dimension $n$. The breakthrough result by~\cite{jin2017escape,jin2019stochastic} achieves iteration complexity $\tilde{O}(\log^{4}(n)/\epsilon^{2})$ which is poly-logarithmic in $n$. The best-known result has complexity $\tilde{O}(\log^{6}(n)/\epsilon^{1.75})$ by~\cite{jin2018accelerated} (the same result in terms of $\epsilon$ was independently obtained by~\cite{allen2018neon2,xu2017neon}). Besides the gradient oracle, escaping from saddle points can also be achieved using the Hessian-vector product oracle with $\tilde{O}(\log(n)/\epsilon^{1.75})$ queries~\cite{agarwal2017finding,carmon2018accelerated}.

There has also been a rich literature on stochastic optimization algorithms for finding second-order stationary points only using the first-order oracle. The seminal work~\cite{ge2015escaping} showed that noisy stochastic gradient descent (SGD) finds approximate second-order stationary points in $O(\poly(n)/\epsilon^{4})$ iterations. This was later improved to $\tilde{O}(\poly(\log n)/\epsilon^{3.5})$~\cite{allen2018natasha2,tripuraneni2018stochastic,fang2019sharp,allen2018neon2,xu2018first}, and the current state-of-the-art iteration complexity of stochastic algorithms is $\tilde{O}(\poly(\log n)/\epsilon^{3})$ due to~\cite{fang2018spider,zhou2019stochastic}.

Quantum algorithms for nonconvex optimization with provable guarantee is a widely open topic. As far as we know, the only work along this direction is by~\cite{zhang2019quantum}, which gives a quantum algorithm for finding the negative curvature of a point in time $\tilde{O}(\poly(r,1/\epsilon))$, where $r$ is the rank of the Hessian at that point. However, the algorithm has a few drawbacks: 1) The cost is expensive when $r=\Theta(n)$; 2) It relies on a quantum data structure~\cite{kerenidis2016recommendation} which can actually be dequantized to classical algorithms with comparable cost~\cite{tang2018quantum2,tang2019quantum,chia2019framework}; 3) It can only find the negative curvature for a fixed Hessian. In all, it is unclear whether this quantum algorithm achieves speedup for escaping from saddle points.

\subsection{Open Questions}
Our paper leaves several natural open questions for future investigation:
\begin{itemize}
\item Can we give quantum-inspired classical algorithms for escaping from saddle points? Our work suggests that compared to uniform perturbation,  there exist physics-motivated methods to better exploit the randomness in gradient descent. A natural question is to understand the potential speedup of using (classical) mechanical waves.
\item Can quantum algorithms achieve speedup in terms of $1/\epsilon$? The current speedup due to quantum simulation can only improve the dependence in terms of $\log n$.
\item Beyond local minima, does quantum provide advantage for approaching global minima? Potentially, simulating quantum wave equations can not only escape from saddle points, but also escape from some poor local minima.
\end{itemize}

\subsection{Organization}
We introduce quantum simulation of the Schr\"{o}dinger equation in \sec{quantum-simulation}, and present how it provides quantum speedup for perturbed gradient descent and perturbed accelerated gradient descent in \sec{quantum-PGD} and \sec{quantum-PAGD}, respectively. We introduce how to replace classical gradient descents by quantum evaluations in \sec{gradient-Jordan}. We present numerical experiments in \sec{numerical}. Necessary tools for our proofs are given in \append{auxiliary-lemmas}.


\section{Escape from Saddle Points by Quantum Simulation}\label{sec:QSimulation}
The main contribution of this section is to show how to escape from a saddle point by replacing the uniform perturbation in the perturbed gradient descent (PGD) algorithm~\cite[Algorithm 4]{jin2019stochastic} and the perturbed accelerated gradient descent (PAGD) algorithm~\cite[Algorithm 2]{jin2018accelerated} with a distribution adaptive to the saddle point geometry. The intuition behind the classical algorithms is that without a second-order oracle, we do not know in which direction a perturbation should be added, thus a uniform perturbation is appropriate. However, quantum mechanics allows us to find the negative curvature direction without explicit Hessian information.

\subsection{Quantum Simulation of the Schr\"{o}dinger Equation}\label{sec:quantum-simulation}
We consider the most standard evolution in quantum mechanics, the Schr\"{o}dinger equation:
\begin{align}\label{eqn:Schrodinger}
i\frac{\partial}{\partial t}\Phi=\Big[-\frac{1}{2}\Delta+f(\vect{x})\Big]\Phi,
\end{align}
where $\Phi$ is a wave function in $\R^{n}$, $\Delta$ is the Laplacian operator, and $f$ can be regarded as the potential of the evolution. In the one-dimensional case, we can prove that $\Phi$ enjoys an explicit form below if $f$ is a quadratic function:
\begin{restatable}{lemma}{standQSimulation}\label{lem:1_dim_standard_QSimulation}
Suppose a quantum particle is in a one-dimensional potential field $f(x)=\frac{\lambda}{2}x^{2}$ with initial state $\Phi(0,x)=(\frac{1}{2\pi})^{1/4}\exp(-x^{2}/4)$; in other words, the initial position of this quantum particle follows the standard normal distribution $\mathcal{N}(0,1)$. The time evolution of this particle is governed by \eqref{eqn:Schrodinger}. Then, at any time $t \ge 0$, the position of the quantum particle still follows normal distribution $\mathcal{N}\left(0,\sigma^2(t;\lambda)\right)$, where the variance $\sigma^2(t;\lambda)$ is given by
\begin{equation} \label{eqn:variance_standard}
     \sigma^2(t;\lambda) = \begin{cases}
     1 + \frac{t^2}{4} & (\lambda = 0),\\
     \frac{(1+4\alpha^2)-(1-4\alpha^2)\cos 2\alpha t}{8\alpha^2} & (\lambda > 0, \alpha = \sqrt{\lambda}),\\
     \frac{(1-e^{2\alpha t})^2 + 4\alpha^2(1+e^{2\alpha t})^2}{16\alpha^2 e^{2\alpha t}} & (\lambda < 0, \alpha = \sqrt{-\lambda}).\\
\end{cases}
\end{equation}
\end{restatable}
\lem{1_dim_standard_QSimulation} shows that the wave function will disperse when the potential field is of negative curvature, i.e., $\lambda < 0$, and the dispersion speed is exponentially fast. Furthermore, we prove in \append{QSimulation-append} that this ``escaping-at-negative-curvature'' behavior of the wave function still emerges given a quadratic potential field $f(\vect{x}) = \frac{1}{2}\vect{x}^T \mathcal{H} \vect{x}$ in any finite dimension.

To turn this idea into a quantum algorithm, we need to use quantum simulation. In fact, quantum simulation in real spaces is a classical problem and has been studied back in the 1990s~\cite{wiesner1996simulations,zalka1998efficient,zalka1998simulating}. There is a rich literature on the cost of quantum simulation~\cite{childs2017note,low2019hamiltonian,berry2015hamiltonian,low2017optimal,lloyd1996universal,berry2007efficient}; it is typically linear in the evolution time, which is formally known as the ``no---fast---forwarding theorem'', see Theorem 3 of~\cite{berry2007efficient}, and Theorem 3 of~\cite{childs2010limitations}. In \sec{simulation}, we prove the following lemma about the cost of simulating the the Schr\"odinger equation using the quantum evaluation oracle in \eqn{quantum-evaluation}:
\begin{lemma}\label{lem:simulation}
Let $f(x)\colon\R^n \to \R$ be a real-valued function with a saddle point at $x=0$ and $f(0) = 0$. Consider the (scaled) Schr\"odinger equation
\begin{align}\label{eqn:Schrodinger-appendix}
	i\frac{\partial}{\partial t}\Phi=\Big[-\frac{r_{0}^2}{2}\Delta+\frac{1}{r_{0}^2}f(\vect{x})\Big]\Phi
\end{align}
defined on the domain $\Omega = \{x\in \R^n: \|x\| \le M\}$ (where $M>0$ is the diameter that will be specified later) with periodic boundary condition.\footnote{Actually, we need to put this $\Omega$ in a flat $n$-torus $\mathbb{T}$, i.e., an $n$-dimensional hybercube with periodic boundary condition, because the flat torus is readily dealt with the finite difference method (FDM). Given the truncation of the function $f(x)$ on $\Omega$, we may slightly ``mollify'' the edge of $f|_{\Omega}$ to observe the periodicity. This mollification will not have a significant impact for optimization because our simulation time is quite short and the wave function rarely has a chance to hit the boundary $\partial \Omega$.} Given the quantum evaluation oracle $U_{f}(\ket{\x}\otimes \ket{0})=\ket{\x}\otimes \ket{f(\x)}$ in \eqn{quantum-evaluation} and an arbitrary initial state at time $t=0$, the evolution of \eqn{Schrodinger-appendix} for time $t>0$ can be simulated using $\tilde{O}\big(t\log n\log^2(\frac{t}{\epsilon})\big)$ queries to $U_{f}$, where $\epsilon$ is the precision.
\end{lemma}

Because we have assumed that $f$ is Hessian-Lipschitz, we can use the second-order Taylor expansion to approximate the function value of $f$ near a saddle point $\tilde{x}$. Such an approximation is more accurate on a ball centered at $\tilde{x}$ with radius $r_{0}$ small enough. Regarding this, we scale the initial distribution as well as the Schr\"{o}dinger equation to be more localized in terms of $r_{0}$, which results in \algo{QuantumSimulation}.

\begin{algorithm}[htbp]
\caption{QuantumSimulation($\tilde{\vect{x}},r_{0},t_{e},f(\cdot)$).}
\label{algo:QuantumSimulation}
Put a Gaussian wave packet into the potential field $f$, with its initial state being:
\begin{align}\label{eqn:ground_state_Phi0}
\Phi_{0}(\vect{x})=\Big(\frac{1}{2\pi}\Big)^{n/4}\frac{1}{r_{0}^{n/2}}\exp(-(\vect{x}-\tilde{\vect{x}})^{2}/4r_{0}^{2});
\end{align}
Simulate its evolution in potential field $f$ with the Schr\"odinger equation for time $t_{e}$\;
Measure the position of the wave packet and output the measurement outcome.
\end{algorithm}

\algo{QuantumSimulation} is the main building block of our quantum algorithms for escaping from saddle points, and also the main resource of our quantum speedup.

\subsubsection{Quantum Query Complexity of Simulating the Schr\"odinger Equation}\label{sec:simulation}
We prove \lem{simulation} in this subsection. Before doing that, we want to briefly discuss the reason why we simulate the scaled Schr\"{o}dinger equation \eqn{Schrodinger-appendix} instead of the common version of non-relativistic Schr\"odinger equation in \eqn{Schrodinger}, rewritten below:
\begin{align}
i\frac{\partial}{\partial t}\Phi=\Big[-\frac{1}{2}\Delta+f(\vect{x})\Big]\Phi.
\end{align}

In real-world problems, we are likely to encounter an objective function $f(x)$ with a saddle point at $x_0$ but is not a quadratic form. In this situation, a quadratic approximation is only valid within a small neighborhood of the first-order stationary point $x_0$, say $\Omega$ defined in \lem{simulation}. Regarding this issue, it is necessary to scale the spatial variable in order to make the wave packet more localized. However, the scaling in the spatial variable will simultaneously cause a scaling in the time variable under Eq.~\eqn{Schrodinger}. This is not preferable because the scaling in time can dramatically change the variance $\sigma(t;\lambda)$ in \eqn{variance_standard}, which can cause troubles when bounding the time complexity in the analysis of algorithms. To leave the time scale invariant, we introduce a modified Schr\"odinger equation \eqn{Schrodinger-appendix}, in which the quantum simulation is restricted on a domain of diameter $O(r_0)$: this localization guarantees that the quantum wave packet captures the saddle point geometry while not to be significantly affected by other features on the landscape of $f(x)$, thus simplifying our further analysis. We may justify our construction of \eqn{Schrodinger-appendix} in three aspects:

\begin{itemize}
    \item \textbf{Geometric aspect:} Eq.~\eqn{Schrodinger-appendix} is obtained by considering a spatial dilation in the wave function $\Phi(t,x) \longmapsto \Phi(t,x/r)$ without changing the time scale. This property guarantees the variance of the Gaussian distribution corresponding to $\Phi(t,x/r)$ is just $r^2$ times the original variance $\sigma^2(t;\lambda)$ (we will prove this in \prop{QSimulation}). Mathematically, this time-invariant property means the dispersion speed is now an intrinsic quantity as it is mostly determined by the saddle point geometry.

    \item \textbf{Physical aspect:} When the wave function is too localized in the position space, due to the uncertainty principle, the momentum variable will spread on a large domain in the frequency space. To reconcile this imparity, we want to introduce a small $r^2$ factor for the kinetic energy operator $-\frac{1}{2}\Delta$ in order to balance between position and momentum.

    \item \textbf{Complexity aspect:} The circuit complexity of simulating Schr\"odinger equation is linear in the operator norm of the Hamiltonian. Our scaling in \eqn{Schrodinger-appendix} drags down the operator norm of the Laplacian (we will discretize it when doing simulation) while leaves the operator norm of the potential field remain $O(\|\mathcal{H}\|)$ in a $O(r_0)$-ball. This normalization effect will help reducing the circuit complexity.
\end{itemize}

Complexity bounds of quantum simulation is a well-established research topic; see e.g.~\cite{childs2017note,low2019hamiltonian,berry2015hamiltonian,low2017optimal,lloyd1996universal,berry2007efficient} for detailed results and proofs. In this paper, we apply quantum simulation under the interaction picture~\cite{low2018hamiltonian}. In particular, we use the following result:
\begin{theorem}[{\cite[Lemma 6]{low2018hamiltonian}}]\label{thm:interaction}
Let $A,B\in\mathbb{C}^{d\times d}$ be time-independent Hamiltonians that are promised to obey $\|A\|\leq \alpha_A$ and $\|B\|\leq\alpha_B$, where $\|\cdot\|$ represents the spectral norm. Then the time-evolution operator $e^{-i(A+B)t}$ can be simulated up to error $\epsilon$ by using
\begin{align*}
    O\Big(\alpha_B t \frac{\log (\alpha_B t/\epsilon)}{\log\log (\alpha_B t/\epsilon)}\Big)
\end{align*}
queries to the unitary oracle $O_B$.\footnote{In fact, Lemma 6 in \cite{low2018hamiltonian} gives an upper bound for the number of queries to the unitary oracle HAM-T. Note that the construction of HAM-T only needs 1 query to $O_B$ (see Theorem 7), we directly give the query complexity in terms of $O_B$.}
\end{theorem}

Our \lem{simulation} is inspired by~\cite{costa2019quantum} which gives a quantum algorithm for simulating the Schr\"{o}dinger equation but without the potential function $f$. It discretizes the space into grids with side-length $a$; in this case, $-\frac{1}{2}\Delta$ reduces to $-\frac{1}{2a^{2}}L$ where $L$ is the Laplacian matrix of the graph of the grids (whose off-diagonal entries are $-1$ for connected grids and zero otherwise; the diagonal entries are the degree of the grids). For instance, in the one-dimensional case,
\begin{align}\label{eqn:spatial-discretization}
-\frac{1}{a^{2}}[L\phi]_{j}=\frac{\phi_{j-1}-2\phi_{j}+\phi_{j+1}}{a^{2}},
\end{align}
where $\phi_{j}$ is the value on the $j^{\text{th}}$ grid. When $a\to 0$, this becomes the second derivative of $\phi$; in practice, as mentioned above, it suffices to take $1/a=\poly(\log(1/\epsilon))$ such that the overall precision is bounded by $\epsilon$.

The discretization method used in~\cite{costa2019quantum} is just a special example of the finite difference method (FDM), which is a common method in applied mathematics to discretize the space of ODE or PDE problems such that their solution is tractable numerically. To be more specific, the continuous space is approximated by discrete grids, and the partial derivatives are approximated by finite differences in each direction. There are higher-order approximation methods for estimating the derivatives by finite difference formulas~\cite{li2005general}, and it is known that the number of grids in each coordinate can be as small as $\poly(\log(1/\epsilon))$ by applying the high-order approximations to the FDM adaptively~\cite{babuvska1987hp}. See also Section 3 of~\cite{childs2020PDE} which gave quantum algorithms for solving PDEs that applied FDM with this $\poly(\log(1/\epsilon))$ complexity for the grids.

We are now ready to prove \lem{simulation}.
\begin{proof}
There are two steps in the quantum simulation of \eqn{Schrodinger-appendix}: (1) discretizing the spatial domain using \eqn{spatial-discretization} so that the Schr\"odinger equation \eqn{Schrodinger-appendix} is reduced to an ordinary differential equation \eqn{Schrodinger-discretization}; (2) simulating \eqn{Schrodinger-discretization} under the interaction picture. In each step, we fix the error tolerance as $\epsilon/2$. By the triangle inequality, the overall error is $\epsilon$.

First, we consider the $k$-th order finite difference method in Section 3 of~\cite{childs2020PDE} (the discrete Laplacian will be denoted as $L_k$). With the spacing between grid points being $a$, if we choose the mesh number along each direction as $1/a=\poly(n)\poly(\log(2/\epsilon))$, the finite difference error will be of order $\epsilon/2$. Then the Schr\"{o}dinger equation in \eqn{Schrodinger-appendix} becomes
\begin{align}\label{eqn:Schrodinger-discretization}
i\frac{\partial}{\partial t}\Phi=\Big(-\frac{r^2_0}{2a^{2}}L_k + B\Big)\Phi,
\end{align}
where $L_k$ is the Laplacian matrix associated to the $k$-th order finite difference method (discretization of the hypercube $\Omega$) and $B$ is a diagonal matrix such that the entry for the grid at $\vect{x}$ is $\frac{1}{r^2_0}f(\vect{x})$. Here, the function evaluation oracle $U_f$ is trivially encoded in the matrix evaluation oracle $O_B$. By~\cite{childs2020PDE}, the spectral norm of $L_k$ is of order $O(n/a^2) = \poly(n)\poly(\log(2/\epsilon))$, where $n$ is the spatial dimension of the Sch\"odinger equation.

We simulate the evolution of \eqn{Schrodinger-discretization} by \thm{interaction} and taking $A=-\frac{r^2_0}{2a^{2}}L_k$ therein. Recall that $\|L_k\| \le \poly(n)\poly(\log(2/\epsilon))$. By the $\ell$-smooth condition, we have $\|\nabla f(\vect{x})\| \le \ell M$ for $\vect{x} \in \Omega$ so that the maximal absolute value of function $f(x)$ on $\Omega$ is bounded by $\ell M^2$ by the Poincar{\'e} inequality. Therefore, we have $\alpha_A \le Cr^2_0\poly(n)\poly(\log(2/\epsilon))$ where $C> 0$ is an absolute constant, and $\alpha_B \le \ell (M/r_0)^2$. It follows from \thm{interaction} that, to simulate the time evolution operator $e^{-i(A+B)t}$ for time $t > 0$, the total quantum queries to $O_B$ (or equivalently, to $U_f$) is
\begin{align*}
O\bigg(\ell (M/r_0)^2 t \big(\log\big( t(Cr^2_0\poly(n)\poly(\log(2/\epsilon)))+\ell (M/r_0)^2 \big)/\epsilon\big)\frac{\log (\ell (M/r_0)^2 \| t/\epsilon)}{\log\log (\ell (M/r_0)^2 t/\epsilon)}\bigg).
\end{align*}
The radius $M$ of the simulation region is chosen large enough such that the wavepacket does not hit the boundary during simulation. Intuitively, the value of $M$ should be proportional to the initial variance $r_0$. Quantitatively, it is shown in \sec{quantum-PGD} such that under our choice of parameters, $M/r_0$ equals some constant $1/C_r$. Absorbing all poly-logarithmic constants in the big $\tilde{O}$ notation, the total quantum queries to $f$ reduces to $\tilde{O}\big(t\log n\log^2(\frac{t}{\epsilon})\big)$ as claimed in \lem{simulation}.
\end{proof}

\begin{remark}
In our scenario of escaping from saddle points, the initial state is a Gaussian wave packet $\big(\frac{1}{2\pi}\big)^{n/4}\frac{1}{r_{0}^{n/2}}\exp(-(\vect{x}-\tilde{\vect{x}})^{2}/4r_{0}^{2})$ as in \algo{QuantumSimulation}. It is well-known that a Gaussian state can be efficiently prepared on quantum computers~\cite{kitaev2008wavefunction}; Gaussian states are also ubiquitous in the literature of continuous-variable quantum information~\cite{weedbrook2012gaussian}. However, although when $f$ is quadratic the evolution of the Schr\"odinger equation keeps the state being a Gaussian wave packet by \lem{standard_multidim_QSimulation}, it intrinsically has dependence on $f$ and it is not totally clear how to prepare the Gaussian wave packet at time $t$ directly by continuous-variable quantum information. It seems that the quantum simulation above using the quantum evaluation oracle $U_{f}$ in \eqn{quantum-evaluation} is necessary for our purpose.
\end{remark}

\subsection{Perturbed Gradient Descent with Quantum Simulation}\label{sec:quantum-PGD}
We now introduce a modified version of perturbed gradient descent. We start with gradient descents until the gradient becomes small. Then, we perturb the point by applying \algo{QuantumSimulation} for a time period $t_{e}=\mathscr{T}'$, perform a measurement on all the coordinates (which gives an output $\vect{x}_{0}$), and continue with gradient descent until the algorithm runs for $T$ iterations. This is summarized as \algo{PGD+QS}.

\begin{algorithm}[htbp]
\caption{Perturbed Gradient Descent with Quantum Simulation.}
\label{algo:PGD+QS}
\For{$t=0,1,...,T$}{
\If{$\|\nabla f(\vect{x}_{t})\|\leq\epsilon$}{
$\xi\sim$QuantumSimulation$\big(\vect{x}_{t},r_{0},\mathscr{T}',f(\vect{x})-\langle\nabla f(\vect{x}_t),\vect{x}-\vect{x}_t\rangle\big)$\label{lin:QuantumSimulation}\;
$\Delta_t\leftarrow\frac{2\xi}{3\|\xi\|}\sqrt{\frac{\rho}{\epsilon}}$\;
$\vect{x}_{t}\leftarrow\mathop{\arg\min}_{\zeta\in\left\{\vect{x}_t+\Delta_t,\vect{x}_t-\Delta_t\right\}}f(\zeta)$\;
}
$\vect{x}_{t+1}\leftarrow\vect{x}_{t}-\eta\nabla f(\vect{x}_{t})$\;
}
\end{algorithm}
Intuitively, in \algo{PGD+QS} QuantumSimulation is applied to find negative curvature of saddle points. Hence in \lin{QuantumSimulation} we simulate the wavepacket under the potential $f(\vect{x})-\langle\nabla f(\vect{x}_t),\vect{x}-\vect{x}_t\rangle$ instead of $f$ itself, since the first order term in the Taylor expansion of $f$ at $\vect{x}_t$ is not relevant to the negative curvature, which is characterized by the second-order Hessian matrix. After negative curvature is specified, we can add a perturbation $\Delta_t$ in that direction to decrease the function value and escape from saddle points.

\subsubsection{Effectiveness of the Perturbation by Quantum Simulation}
We show that our method of quantum wave packet simulation is significantly better than the classical method of uniform perturbation in a ball. To be more specific, we focus on the scenarios with $\epsilon\leq l^{2}/\rho$ (this is the standard assumption adopted in~\cite{jin2018accelerated}); intuitively, this is the case when the local landscape is ``flat'' and the Hessian has a small spectral radius. Under this circumstance, the classical gradient descent may move slowly, but the quantum Gaussian wavepacket still disperses fast, i.e., the variance of the probability distribution corresponding to the wavepacket still has a large increasing rate. Hence, if we let this wavepacket evolve for a long enough time period, it is drastically stretched in the directions with negative curvature. As a result, if we measure its position at this time, with high probability the output vector indicates a negative curvature direction, or equivalently, a direction along which we can decrease the function value. We can thus add a large perturbation along that direction to escape from the saddle point. Formally, we prove:

\begin{restatable}{proposition}{QSeffectiveness}\label{prop:QS-effectiveness}
For any constant $\delta_0>0$, we specify our choices for the parameters and constants that we use:
\begin{align}
    \mathscr{T}'   &:= \frac{8}{(\rho\epsilon)^{1/4}}\log\Big(\frac{\ell}{\delta_0\sqrt{\rho\epsilon}}(n+2\log(3/\delta_0))\Big)     & \mathscr{F}' &:= \frac{2}{81}\sqrt{\frac{\epsilon^3}{\rho}}
    \label{eqn:parameter-choice-1}\\
    r_0  &:= \frac{4C_r^3}{9\mathscr{T}'^4}\Big(\frac{\delta_0}{3}\cdot\frac{1}{n^{3/2}+2C_0n\ell(\log \mathscr{T}')^{\alpha}}\Big)^2
    & \eta  &:= \frac{1}{\ell}
    \label{eqn:parameter-choice-2}
\end{align}
where $C_r$, $C_0$, $\alpha$ are absolute constants, and the value of $\alpha$ is specified in \lem{deviation-from-quadratic}. Let $f\colon \R^n \to \R$ be an $\ell$-smooth, $\rho$-Hessian Lipschitz function. For an approximate saddle point $\vect{\tilde{x}}$ satisfying $\|\nabla f(\vect{\tilde{x}})\| \leq \epsilon$ and $\lambda_{\min}(\nabla^{2}f(\vect{\tilde{x}}))\leq -\sqrt{\rho \epsilon}$, \algo{PGD+QS} adds a perturbation by QuantumSimulation with the radius $M$ of the simulation region being set to $M=r_0/C_r$, and decreases the function value for at least $\mathscr{F}'$, with probability at least $1-\delta_0$.
\end{restatable}

Compared to the provable guarantee from classical perturbation~\cite[Lemma 22]{jin2019stochastic}, speaking only in terms of $n$, classically it takes $\mathscr{T}=O(\log n)$ steps to decrease the function value by $\mathscr{F}=\Omega(1/\log^3 n)$, whereas our quantum simulation with time $\mathscr{T}'=O(\log n)$ together with also $\mathscr{T}'$ subsequent GD iterations decrease the function value by $\mathscr{F}'=\Omega(1)$ with high success probability.

Intuitively, the proof of \prop{QS-effectiveness} is composed of two parts. If the potential $f$ is quadratic, we can use \lem{1_dim_standard_QSimulation} to prove \prop{QSimulation} (both proof details are given in \append{QSimulation-append}), which demonstrates the exponential rate for quantum simulation to escape along the negative eigen-directions of the Hessian of $f$. However, the objective function $f$ is rarely a standard quadratic form in reality, and we cannot expect the quantum wave packet to preserve its shape as a Gaussian distribution. Nevertheless, we are able to show that the quantum wave packets do not differ significantly from a perfect Gaussian distribution in the course of quantum simulation, which preserves our quantum speedup in the general case.

Formally, we introduce the following lemma to bound the deviation from perfect Gaussian in quantum evolution. Before proceeding to its details, we first specify our choice for the constant $C_r$. As shown in the statement of \prop{QS-effectiveness}, $C_r$ stands for the ratio between the initial wavepacket variance and the radius of the simulation region. We choose a small enough constant $C_r$, such that the simulation region would be much larger than the range of the wavepacket, during the entire simulation process. Since the function $f$ is $\ell$-smooth, the spectral norm of its Hessian matrix at any point is upper bounded by constant $\ell$. Hence, the small enough constant $C_r$ is independent of $f$. Then, the radius $M$ of the simulation region satisfies
\begin{align}\label{eqn:M}
    M=r_{0}/C_r=\frac{4C_r^2}{9\mathscr{T}'^4}\Big(\frac{\delta_0}{3}\cdot\frac{1}{n^{3/2}+2C_0n\ell(\log \mathscr{T}')^{\alpha}}\Big)^2\leq 1.
\end{align}

\begin{restatable}{lemma}{DeviationFQuadratic}\label{lem:deviation-from-quadratic}
		Let $\mathcal{H}$ be the Hessian matrix of $f$ at a saddle point $\tilde{\vect{x}}$, and define $f_{q}(\vect{x}) := f(\tilde{\vect{x}}) + \frac{1}{2} (\vect{x} - \tilde{\vect{x}})^T \mathcal{H} (\vect{x} - \tilde{\vect{x}})$ to be the quadratic approximation of the function $f$ near $\tilde{\vect{x}}$. Denote the measurement outcome from the quantum simulation (see \algo{QuantumSimulation}) with potential field $f$ and evolution time $t_e$ as random variable $\xi$, and the measurement outcome from the quantum simulation with potential field $f_{q}$ and the same evolution time $t_e$ as another random variable $\xi'$. Let the law of $\xi$ (or $\xi'$, resp.) be $\mathbb{P}_{\xi}$  (or $\mathbb{P}_{\xi'}$, resp.). If the quantum wave packet is confined to a hypercube with edge length $M$, then
		\begin{align}
			TV(\mathbb{P}_{\xi}, \mathbb{P}_{\xi'}) \le \left(\frac{\sqrt{n}\rho}{2} + \frac{2C_f \ell}{\sqrt{r_0}} (\log t_e)^\alpha \right)\frac{n M t^2_e}{2},
		\end{align}
		where $TV(\cdot, \cdot)$ is the total variation distance between measures, $\alpha$ is an absolute constant, and $C_f$ is an $f$-related constant.
\end{restatable}

The proof of \lem{deviation-from-quadratic} is deferred to \append{evol_deviation}. This lemma shows that the true perturbation given by quantum simulation $\xi \sim \mathbb{P}_{\xi}$ only deviates from the Gaussian random vector $\xi' \sim \mathbb{P}_{\xi'}$ at a magnitude of $\tilde{O}(Mn^{3/2}t^2_e)$ when $t_{e}=\mathscr{T}'= O(\log n)$ in \algo{PGD+QS}.  Such a deviation is non-material compared to our choice of $M$ in \eqn{M}. Therefore, we may estimate the performance of our quantum simulation subroutine using a quadratic approximation function and then bound the error caused by the non-quadratic part, as in the following lemma:

\begin{lemma}\label{lem:negative-curvature}
Under the setting of \prop{QS-effectiveness}, let $\tilde{\mathcal{H}}$ be the Hessian matrix of $f$ at point $\tilde{\vect{x}}$. Then, the output of QuantumSimulation$(\tilde{\vect{x}},r_0,\mathscr{T}')$ by applying \algo{QuantumSimulation}, denoted as $\xi$, satisfies
\begin{align}
    \frac{\xi^T \tilde{\mathcal{H}}\xi}{\|\xi\|^2}\leq -\frac{\sqrt{\rho\epsilon}}{3},
\end{align}
with probability at least $1-\delta_0$.
\end{lemma}
\begin{proof}
Without loss of generality, assume $\nabla f(\vect{x}_t)=\vect{0}$. First consider the case where the potential $f$ is purely quadratic, and add the estimate the error term caused by the non-quadratic deflation afterwards.

First note that the Hessian matrix $\tilde{\mathcal{H}}$ admits the following eigen-decomposition:
\begin{align}\label{eqn:Hessian-decomposition}
\tilde{\mathcal{H}}=\sum_{i=1}^{n}\lambda_i\vect{u}_{i}\vect{u}_i^{T},
\end{align} where the set $\{\vect{u}_i\}_{i=1}^{n}$ forms an orthonormal basis of $\mathbb{R}^n$. Without loss of generality, we assume that the eigenvalues $\lambda_1,\lambda_2,\ldots,\lambda_n$ corresponding to $\vect{u}_1,\vect{u}_2,\ldots,\vect{u}_n$ satisfy
\begin{align}
\lambda_1\leq\lambda_2\leq\cdots\leq\lambda_n,
\end{align}
in which $\lambda_1\leq-\sqrt{\rho\epsilon}$. If $\lambda_n\leq-\sqrt{\rho\epsilon}/2$, \lem{negative-curvature} holds directly. Hence, we only need to prove the case where $\lambda_n>-\sqrt{\rho\epsilon}/2$, in which there exists some $p$, $p'$ with
\begin{align}
\lambda_p\leq -\sqrt{\rho\epsilon}< \lambda_{p+1},\quad\lambda_{p'}\leq -\sqrt{\rho\epsilon}/2< \lambda_{p'+1}.
\end{align}
We use $\mathfrak{S}_{\parallel}$, $\mathfrak{S}_{\perp}$ to respectively denote the subspace of $\mathbb{R}^{n}$ spanned by
\begin{align}
\left\{\vect{u}_1,\vect{u}_2,\ldots,\vect{u}_p\right\}, \quad
\left\{\vect{u}_{p+1},\vect{u}_{p+2}, \ldots, \vect{u}_{n}\right\},
\end{align}
and use $\mathfrak{S}_{\parallel}'$, $\mathfrak{S}_{\perp}'$ to respectively denote the subspace of $\mathbb{R}^{n}$ spanned by
\begin{align}
    \left\{\vect{u}_1,\vect{u}_2,\ldots,\vect{u}_{p'}\right\},\quad \left\{\vect{u}_{p'+1},\vect{u}_{p+2},\ldots,\vect{u}_{n}\right\}.
\end{align}
Furthermore, we define $\xi_{\parallel}:=\sum_{i=1}^p\langle\vect{u}_i,\xi\rangle\vect{u}_i$,
$\xi_{\perp}:=\sum_{i=p}^n\langle\vect{u}_i,\xi\rangle\vect{u}_i$,
$\xi_{\parallel'}:=\sum_{i=1}^{p'}\langle\vect{u}_i,\xi\rangle\vect{u}_i$,
$\xi_{\perp'}:=\sum_{i=p'}^n\langle\vect{u}_i,\xi\rangle\vect{u}_i$
respectively to denote the component of $\xi$ in the subspaces $\mathfrak{S}_{\parallel}$, $\mathfrak{S}_{\perp}$, $\mathfrak{S}_{\parallel}'$, $\mathfrak{S}_{\perp}'$. Also, we define $\xi_1:=\langle\vect{u}_1,\xi\rangle\vect{u}_1$ to be the component of $\xi$ along $\vect{u}_1$, the most negative eigen-direction.

Under the basis $\left\{\vect{u}_1,\ldots,\vect{u}_n\right\}$, by \prop{QSimulation}, the time evolution of the initial wave function is governed by \eqref{eqn:Schrodinger-appendix}. Then, at $t_e=\mathscr{T}'$, the wave function still follows multivariate Gaussian distribution $\mathcal{N}(0, r_0^2\Sigma)$, with the covariance matrix being
\begin{align}
    \Sigma =  \diag(\sigma^2(\mathscr{T}';\lambda_1), ..., \sigma^2(\mathscr{T}';\lambda_n)),
\end{align}
where the variance $\sigma(\mathscr{T}';\lambda_i)$ is defined in \eqn{variance_standard}. Denote $\sigma_i:=\sigma(\mathscr{T}';\lambda_i)$ for each $i\in[n]$. Then, for any $i\in[n]$ with $\vect{u}_i\in\mathfrak{S}_{\perp}'$, since $\lambda_i\geq-\sqrt{\rho\epsilon}/2$, we have
\begin{align}
    \sigma_i^2\leq\frac{(1-e^{(4\rho\epsilon)^{1/4}\mathscr{T}'})^2 + \rho\epsilon(1+e^{(4\rho\epsilon)^{1/4}\mathscr{T}'})^2}{4\rho\epsilon \cdot e^{(4\rho\epsilon)^{1/4}\mathscr{T}'}}.
\end{align}
Due to our choice of the parameter $\mathscr{T}'$, we can further derive that
\begin{align}
    \sigma_i^2
    \leq\frac{(1+2\rho\epsilon) e^{2(4\rho\epsilon)^{1/4}\mathscr{T}'}}{4\rho\epsilon \cdot e^{(4\rho\epsilon)^{1/4}\mathscr{T}'}}
    \leq\frac{ e^{(4\rho\epsilon)^{1/4}\mathscr{T}'}}{2\rho\epsilon}.
\end{align}
Denote $\sigma_{\perp}':=\frac{ e^{(4\rho\epsilon)^{1/4}\mathscr{T}'}}{2\rho\epsilon}$. We define an $(n-p')$-dimensional Gaussian distribution $\tilde{p}(\cdot)$ in $\mathfrak{S}_{\perp}'$:
\begin{align}
    \tilde{p}(\vect{y})=\Big(\frac{1}{2\pi}\Big)^{(n-p')/2}\Big(\frac{\sqrt{n-p'}}{\sigma_{\perp}'r_0}\Big)\exp\Big(-\frac{(n-p')\|\vect{y}\|^2}{2\sigma_{\perp}'{}^2 r_0^2}\Big),
\end{align}
then the actual distribution of $\|\xi_{\perp'}\|$ is upper bounded by the distribution of $\|\vect{y}\|$ under the probability density function $\tilde{p}(\vect{y})$. Furthermore, by \lem{positive-tail} in \append{h-tail}, with probability at least $1-\delta_0/3$ we have
\begin{align}
    \|\xi_{\perp'}\|^2/r_0^2
    &\leq \sum_{i=p'+1}^{n}\sigma_i^2+2\sqrt{\log(3/\delta_0)\sum_{i=p'+1}^{n}\sigma_i^4}+2\max_{p'+1\leq i\leq n}\sigma_i^2 \log(3/\delta_0)\\
    &\leq (n-p')\sigma_{\perp}'{}^2+2\big(\sqrt{(n-p')\log(3/\delta_0)}+\log(3/\delta_0)\big)\sigma_{\perp}'{}^2\\
    &\leq 2(n+2\log(3/\delta_0))\sigma_{\perp}'{}^2.
\end{align}

On the other hand, on the most negative direction $i=1$, by $\lambda_1\leq-\sqrt{\rho\epsilon}$, we can derive that
\begin{align}
    \sigma_1^2
    &\geq\frac{(1-e^{2(\rho\epsilon)^{1/4}\mathscr{T}'})^2 + 4\rho\epsilon(1+e^{2(\rho\epsilon)^{1/4}\mathscr{T}'})^2}{16\rho\epsilon e^{2(\rho\epsilon)^{1/4}\mathscr{T}'}}\\
    &\geq\frac{e^{4(\rho\epsilon)^{1/4}\mathscr{T}'}/2 + 4\rho\epsilon e^{4(\rho\epsilon)^{1/4}\mathscr{T}'}}{16\rho\epsilon e^{2(\rho\epsilon)^{1/4}\mathscr{T}'}}\\
    &\geq \frac{e^{2(\rho\epsilon)^{1/4}\mathscr{T}'}}{32\rho\epsilon}.
\end{align}
Hence, after we measure the wavepacket, $\xi_1$ satisfies
\begin{align}
    \Pr\left\{|\xi_1|\geq \frac{\delta_0\sigma_1 r_0}{2} \right\}&=\int_{-\delta_0\sigma_1 r_0/2}^{\delta_0\sigma_1 r_0/2}\Big(\frac{1}{2\pi}\Big)^{1/2}\cdot\frac{1}{r_0\sigma_1}\exp\Big(-\frac{\theta^2}{2r_0^2\sigma_1^2}\Big)\d \theta\\
    &\geq\Big(\frac{1}{2\pi}\Big)^{1/2}\cdot\frac{2}{r_0\sigma_1}\cdot\frac{\delta_0\sigma_1 r_0}{2}\geq\frac{\delta_0}{3}.
\end{align}
By the union bound, with probability at least $1-2\delta_0/3$, the output $\xi$ would satisfy:
\begin{align}
    \frac{\|\xi_{\perp'}\|}{\|\xi_{\parallel'}\|}&\leq\frac{\|\xi_{\perp}\|}{|\xi_1|}
    \leq \frac{\sqrt{2(n+2\log(3/\delta_0))}}{\delta_0/2}\cdot\frac{\sigma_{\perp'}'}{\sigma_1}\\
    &\leq \frac{3\sqrt{(n+2\log(3/\delta_0))}}{\delta_0}\cdot\frac{ e^{(4\rho\epsilon)^{1/4}\mathscr{T}'/2}}{\sqrt{2\rho\epsilon}}\cdot\frac{\sqrt{32\rho\epsilon}}{e^{(\rho\epsilon)^{1/4}\mathscr{T}'}}\\
    &\leq\frac{12\sqrt{(n+2\log(3/\delta_0))}}{\delta_0}\cdot \exp\big(-(1-\sqrt{2}/2)(\rho\epsilon)^{1/4}\mathscr{T}'\big)\\
    &\leq \frac{\sqrt{\rho\epsilon}}{12\ell}.
\end{align}
Considering the fact that the function $f$ is not purely quadratic, by \lem{deviation-from-quadratic} the inequality above may be violated with probability at most
\begin{align}
    \frac{2}{3}\delta_0+TV(\mathbb{P}_{\xi}, \mathbb{P}_{\xi'}) \leq \frac{2}{3}\delta_0+ \left( \sqrt{n}\rho + \frac{2C \ell}{\sqrt{r_0}} (\log \mathscr{T}')^\alpha \right)\frac{n M \mathscr{T}'^2}{2},
\end{align}
in which $M=r_0/C_r$ due to our parameter choice. Choose the constant $C_0$ in $r_0$ large enough to satisfy $C_0\geq C$. Then with probability at least $1-\delta_0$, we can still have
\begin{align}
    \frac{\|\xi_{\perp'}\|}{\|\xi_{\parallel'}\|}\leq \frac{\sqrt{\rho\epsilon}}{12\ell},
\end{align}
after counting in the deviation from pure quadratic field. Under this circumstance, use $\hat{\xi}$ to denote $\xi/\|\xi\|$. Observe that
\begin{align}
\hat{\xi}^{T}\tilde{\mathcal{H}}\hat{\xi}=(\hat{\xi}_{\perp'}+\hat{\xi}_{\parallel'})^{T}\tilde{\mathcal{H}}(\hat{\xi}_{\perp'}+\hat{\xi}_{\parallel'})=\hat{\xi}_{\perp'}^{T}\tilde{\mathcal{H}}\hat{\xi}_{\perp'}+\hat{\xi}_{\parallel'}^{T}\tilde{\mathcal{H}}\hat{\xi}_{\parallel'}
\end{align}
since $\tilde{\mathcal{H}}\hat{\xi}_{\perp'}\in\mathfrak{S}_{\perp}'$ and $\tilde{\mathcal{H}}\hat{\xi}_{\parallel'}\in\mathfrak{S}_{\parallel}'$. Due to the $\ell$-smoothness of the function, all eigenvalue of the Hessian matrix has its absolute value upper bounded by $\ell$. Thus we have,
\begin{align}
\hat{\xi}_{\perp'}^{T}\tilde{\mathcal{H}}\hat{\xi}_{\perp'}\leq\ell\|\hat{\xi}_{\perp'}^{T}\|_{2}^2=\rho\epsilon/(144\ell^2).
\end{align}
Further according to the definition of $\mathfrak{S}_{\parallel}$, we have
\begin{align}
\hat{\xi}_{\parallel'}^{T}\tilde{\mathcal{H}}\hat{\xi}_{\parallel'}\leq-\sqrt{\rho\epsilon}\|\hat{\xi}_{\parallel'}\|^2/2.
\end{align}
Combining these two inequalities together, we can obtain
\begin{align}
\hat{\xi}^{T}\tilde{\mathcal{H}}\hat{\xi}&=\hat{\xi}_{\perp'}^{T}\tilde{\mathcal{H}}\hat{\xi}_{\perp'}+\hat{\xi}_{\parallel'}^{T}\tilde{\mathcal{H}}\hat{\xi}_{\parallel'}
\leq-\sqrt{\rho\epsilon}\|\hat{\xi}_{\parallel'}\|^2/2+\rho\epsilon/(144\ell^2)\leq-\sqrt{\rho\epsilon}/3.
\end{align}
\end{proof}

Now we are ready to prove \prop{QS-effectiveness}.
\begin{proof}
Without loss of generality, we assume $\tilde{\vect{x}}=\vect{0}$. By \lem{negative-curvature}, with probability at least $1-\delta_0$, the output $\xi$ of QuantumSimulation would be in a negative curvature direction, or quantitatively,
\begin{align}
    \frac{\xi^T \tilde{\mathcal{H}}\xi}{\|\xi\|^2}\leq -\sqrt{\rho\epsilon}/3.
\end{align}
Since we choose the one with smaller function value from $\left\{\Delta_t,-\Delta_t\right\}$ to be the perturbation result, without loss of generality we can assume $\langle \nabla f(\vect{0}), \Delta_t\rangle\leq 0$. Then,
\begin{align}
    f(\Delta_t)-f(\vect{0})&=\int_{0}^{1}\langle \nabla f(\theta\Delta_t), \Delta_t\rangle \d \theta,
\end{align}
where the gradient $\nabla f(\theta\Delta_t)$ can be expressed as
\begin{align}
    \nabla f(\theta\Delta_t)=\nabla f(\vect{0})+\int_{0}^{\theta}\mathcal{H}(\nu\Delta_t)\Delta_t\d \nu,
\end{align}
which leads to
\begin{align}
    f(\Delta_t)-f(\vect{0})&=\langle\nabla f(\vect{0}),\Delta_t\rangle+\int_{0}^{1}\d\theta\Big\langle\int_{0}^{\theta}\mathcal{H}(\nu\Delta_t)\Delta_t\d \nu,\Delta_t\Big\rangle\\
    &\leq\int_{0}^{1}\d\theta\int_{0}^{\theta}\langle\mathcal{H}(\nu\Delta_t)\Delta_t,\Delta_t\rangle\d\nu.
\end{align}
Here, $\mathcal{H}(\nu,\Delta_t)$ satisfies
\begin{align}
    \|\mathcal{H}(\nu\Delta_t)-\tilde{\mathcal{H}}\|\leq \rho\|\nu\Delta_t\|
\end{align}
due to the $\rho$-Hessian Lipschitz property of $f$, which indicates
\begin{align}
    \langle\mathcal{H}(\nu\Delta_t)\Delta_t,\Delta_t\rangle
    &=\langle\tilde{\mathcal{H}}\Delta_t,\Delta_t\rangle+\langle(\mathcal{H}(\nu\Delta_t)-\tilde{\mathcal{H}})\Delta_t,\Delta_t\rangle\\
    &\leq\langle\tilde{\mathcal{H}}\Delta_t,\Delta_t\rangle +\|\mathcal{H}(\nu\Delta_t)-\tilde{\mathcal{H}}\|\cdot\|\Delta_t\|^2\\
    &\leq\langle\tilde{\mathcal{H}}\Delta_t,\Delta_t\rangle +\rho\|\Delta_t\|^3\nu,\quad\forall \nu>0.
\end{align}
Hence,
\begin{align}
    f(\Delta_t)-f(\vect{0})&\leq\int_{0}^{1}\d\theta\int_{0}^{\theta}\langle\mathcal{H}(\nu\Delta_t)\Delta_t,\Delta_t\rangle\d\nu\\
    &\leq\int_{0}^{1}\d\theta\int_{0}^{\theta}\langle\tilde{\mathcal{H}}\Delta_t,\Delta_t\rangle\d\nu+\int_{0}^{1}\d\theta\int_{0}^{\theta}\rho\|\Delta_t\|^3\nu\d\nu\\
    &\leq-\frac{\sqrt{\rho\epsilon}}{6}\cdot\|\Delta_t\|^2+\frac{\rho}{6}\cdot\|\Delta_t\|^3\\
    &=-\frac{\sqrt{\rho\epsilon}}{6}\cdot\frac{4\epsilon}{9\rho}+\frac{\rho}{6}\cdot\frac{8\epsilon^{3/2}}{27\rho^{3/2}}=-\mathscr{F}'.
\end{align}
\end{proof}

\subsubsection{Proof of Our Quantum Speedup}
We now prove the following theorem using \prop{QS-effectiveness}:
\begin{restatable}{theorem}{PGDQScomplexity}\label{thm:PGD+QS-Complexity}
For any $\epsilon$, $\delta>0$, \algo{PGD+QS} with parameters chosen in \prop{QS-effectiveness} satisfies that at least one half of its iterations of will be $\epsilon$-approximate local minima, using
\begin{align*}
\tilde{O}\Big(\frac{(f(\vect{x}_{0})-f^{*})}{\epsilon^{2}}\cdot\log^{2}n\Big)
\end{align*}
queries to $U_{f}$ in \eqn{quantum-evaluation} and gradients with probability $\geq1-\delta$, where $f^{*}$ is the global minimum of $f$.
\end{restatable}

\begin{proof}
Set $\delta_0=\frac{2}{81(f(\vect{x}_0)-f^{*})}\sqrt{\frac{\epsilon^3}{\rho}}$, let the parameters be chosen according to \eqn{parameter-choice-1} and \eqn{parameter-choice-2}, and set the total iteration steps $T$ to be:
\begin{align}
T=4\max\left \{ \frac{(f(\vect{x_{0}})-f^{*})}{\mathscr{F'}},\frac{(f(\vect{x_{0}})-f^{*})}{\eta \epsilon^{2}} \right \}
=\tilde{O}\Big(\frac{(f(\vect{x}_{0})-f^{*})}{\epsilon^{2}}\cdot\log n\Big),
\end{align}
similar to the classical GD algorithm.  We first assume that for each $\vect{x}_t$ we apply QuantumSimulation (\algo{QuantumSimulation}),
we can successfully obtain an output $\xi$ with $\xi^{T}\mathcal{H}\xi/\|\xi\|^2\leq-\sqrt{\rho\epsilon}/3$, as long as $\lambda_{\min}(\mathcal{H}(\vect{x}_t))\leq-\sqrt{\rho\epsilon}$. The error probability of this assumption is provided later.

Under this assumption, \algo{QuantumSimulation} can be called for at most $\frac{81(f(\vect{x_{0}})-f^{*})}{2}\sqrt{\frac{\rho}{\epsilon^3}}\leq \frac{T}{4}$ times, for otherwise the function value decrease will be greater than $f(\vect{x_{0}})-f^{*}$, which is not possible. Then, the error probability that some calls to \algo{QuantumSimulation} fail to indicate a negative curvature is upper bounded by
\begin{align}
\frac{81(f(\vect{x_{0}})-f^{*})}{2}\sqrt{\frac{\rho}{\epsilon^3}}\cdot\delta_0=\delta.
\end{align}

Excluding those iterations that QuantumSimulation is applied, we still have at least $3T/4$ steps left. They are either large gradient steps, $\|\nabla f(\vect{x}_{t})\|\geq \epsilon$, or $\epsilon$-approximate second-order stationary points. Within them, we know that the number of  large gradient steps cannot be more than $T/4$ because otherwise, by \lem{descent-lemma} in \append{existing-lemma}:
\begin{align}
f(\vect{x}_{T})\leq f(\vect{x}_{0})-T\eta\epsilon^{2}/4<f^{*},
\end{align}
a contradiction. Therefore, we conclude that at least $T/2$ of the iterations must be $\epsilon$-approximate second-order stationary points with probability at least $1-\delta$.

The number of queries can be viewed as the sum of two parts, the number of queries needed for gradient descent, denoted by $T_{1}$, and the number of queries needed for quantum simulation, denoted by $T_{2}$. Then with probability at least $1-\delta$,
\begin{align}
T_{1}=T=\tilde{O}\Big(\frac{(f(\vect{x}_{0})-f^{*})}{\epsilon^{2}}\cdot\log n\Big).
\end{align}
As for $T_{2}$, with probability at least $1-\delta$ quantum simulation is called for at most $\frac{4(f(\vect{x_{0}})-f^{*})}{\mathscr{F'}}$ times, and by \lem{simulation} it takes $\tilde{O}\big(\mathscr{T}'\log n \log^{2}(\mathscr{T}'^{2}/\epsilon)\big)$ queries to carry out each simulation. Therefore,
\begin{align}
T_{2}=\frac{4(f(\vect{x_{0}})-f^{*})}{\mathscr{F'}}\cdot\tilde{O}\big(\mathscr{T}'\log n \log^{2}(\mathscr{T}'^{2}/\epsilon)\big)=\tilde{O}\Big(\frac{(f(\vect{x}_{0})-f^{*})}{\epsilon^{1.75}}\cdot\log^2 n\Big).
\end{align}
As a result, the total query complexity $T_{1}+T_{2}$ is
\begin{align}
\tilde{O}\Big(\frac{(f(\vect{x}_{0})-f^{*})}{\epsilon^{2}}\cdot\log^2 n\Big).
\end{align}
\end{proof}

\subsection{Perturbed Accelerated Gradient Descent with Quantum Simulation}\label{sec:quantum-PAGD}
In \thm{PGD+QS-Complexity}, the $1/\epsilon^{2}$ term is a bottleneck of the whole algorithm, but~\cite{jin2018accelerated} improved it  to $1/\epsilon^{1.75}$ by replacing the GD with the accelerated GD by~\cite{nesterov1983method}. We next introduce a hybrid quantum-classical algorithm (\algo{PAGDQS}) that reflects this intuition. We make the following comparisons to~\cite{jin2018accelerated}:
\begin{itemize}
\item \textbf{Same:} When the gradient is large, we both apply AGD iteratively until we reach a point with small gradient. If the function becomes ``too nonconvex'' in the AGD, we both reset the momentum and decide whether to exploit the negative curvature at that point.

\item \textbf{Difference:} At a point with small gradient, we apply quantum simulation instead of the classical uniform perturbation. Speaking only in terms of $n$,~\cite{jin2018accelerated} takes $O(\log n)$ steps to decrease the Hamiltonian $f(\vect{x})+\frac{1}{2\eta}\|\vect{v}\|^{2}$ by $\Omega(1/\log^5 n)$ with high probability, whereas our quantum simulation for time $\mathscr{T}'=O(\log n)$ decreases the Hamiltonian by $\Omega(1)$ with high probability.
\end{itemize}

\begin{algorithm}[htbp]
\caption{Perturbed Accelerated Gradient Descent with Quantum Simulation.}
\label{algo:PAGDQS}
$\vect{v}_{0}\leftarrow 0$\;
\For{$t=0,1,\ldots,T$}{
\If{$\|\nabla f(\vect{x}_{t})\|\leq \epsilon$}{
$\xi\sim$QuantumSimulation$\big(\vect{x}_{t},r_{0},\mathscr{T}',f(\vect{x})-\langle\nabla f(\vect{x}_t),\vect{x}-\vect{x}_t\rangle\big)$\;
$\Delta_t\leftarrow\frac{2\xi}{3\|\xi\|}\sqrt{\frac{\rho}{\epsilon}}$\;
$\vect{x}_{t}\leftarrow\mathop{\arg\min}_{\zeta\in\left\{\vect{x}_t+\Delta_t,\vect{x}_t-\Delta_t\right\}}f(\zeta)$\;
$\vect{v}_t\leftarrow\vect{0}$\;
}
\Else{
$\vect{y}_{t}\leftarrow \vect{x}_{t}+(1-\theta)\vect{v}_{t}$, $\vect{x}_{t+1}\leftarrow \vect{y}_{t}-\eta' f(\vect{y}_{t})$, and $\vect{v}_{t+1}\leftarrow \vect{x}_{t+1}-\vect{x}_{t}$\;
\If{$f(\vect{x}_{t})\leq f(\vect{y}_{t})+\left \langle \nabla f(\vect{y}_{t}),\vect{x}_{t}-\vect{y}_{t} \right \rangle -\frac{\gamma}{2}\|\vect{x}_{t}-\vect{y}_{t}\|$}{
$(\vect{x}_{t+1},\vect{v}_{t+1})\leftarrow$Negative-Curvature-Exploitation$(\vect{x}_{t},\vect{v}_{t},s)$\label{lin:NCE}\;
}
}
}
\end{algorithm}

The following theorem provides the complexity of this algorithm:
\begin{theorem}\label{thm:QuantumSimulationAGD}
Suppose that the function $f$ is $\ell$-smooth and $\rho$-Hessian Lipschitz. We choose the parameters appearing in \algo{PAGDQS} as follows:
\begin{align}
    \delta_0 &:=\frac{2}{81(f(\vect{x}_0)-f^{*})}\sqrt{\frac{\epsilon^3}{\rho}} &\chi&:=1   & \eta  &:= \frac{1}{\ell}\\
    \mathscr{T}'   &:= \frac{8}{(\rho\epsilon)^{1/4}}\log\Big(\frac{\ell}{\delta_0\sqrt{\rho\epsilon}}(n+2\log(3/\delta_0))\Big)
    &\eta'       &:=\frac{1}{4\ell}
    & \mathscr{F}' &:= \frac{2}{81}\sqrt{\frac{\epsilon^3}{\rho}} \\
    r_{0}  &:= \frac{4C_r^3}{9\mathscr{T}'^4}\Big(\frac{\delta_0}{3}\cdot\frac{1}{n^{3/2}+2C_0n\ell(\log \mathscr{T}')^{\alpha}}\Big)^2     &\kappa & := \frac{\ell}{\sqrt{\rho\epsilon}}   & \theta &:=\frac{1}{4\sqrt{\kappa}}\\         \gamma &:=\frac{\theta^{2}}{\eta}     &s&:=\frac{\gamma}{4\rho}     &\mathscr{T}&:=\sqrt{\kappa}\cdot c_{A}
\end{align}
where $c_{A}$ is chosen large enough to satisfy the condition in \lem{AGD-large-gradient}, $C_0$ and $C_r$ are constants specified in \prop{QS-effectiveness}. Then, for any $\delta>0$, $\epsilon\leq\frac{\ell^{2}}{\rho}$, if we run \algo{PAGDQS} with choice of parameters specified above, then with probability at least $1-\delta$ one of the iterations $\vect{x}_{t}$ will be an $\epsilon$-approximate second-order stationary point, using the following number of queries to $U_{f}$ in \eqn{quantum-evaluation} and classical gradients:
\begin{align}
\tilde{O}\Big(\frac{(f(\vect{x}_{0})-f^{*})}{\epsilon^{1.75}}\cdot \log^{2}n\Big).
\end{align}
\end{theorem}

\begin{proof}
We use $T$ to denote total number of iterations and specify our choice for $T$ as:
\begin{align}
T=3\max\left \{ \frac{(f(\vect{x_{0}})-f^{*})}{\mathscr{F'}},\frac{(f(\vect{x_{0}})-f^{*})\mathscr{T}}{\mathscr{E}} \right \},
\end{align}
where $\mathscr{E}=\sqrt{\frac{\epsilon^{3}}{\rho}}\cdot c_{A}^{-7}$, the same as our choice for $\mathscr{E}$ in \lem{AGD-large-gradient}. Similar to \prop{QS-effectiveness}, we set the radius $M$ of the simulation region to be $r_0/C_r$. We assume the contrary, i.e., the outputs of all of the iterations are not $\epsilon$-approximate second-order stationary points.

Similar to our analysis in the proof of \thm{PGD+QS-Complexity}, we first assume that for each $\vect{x}_t$ we apply QuantumSimulation (\algo{QuantumSimulation}),
we can successfully obtain an output $\xi$ with $\xi^{T}\mathcal{H}\xi/\|\xi\|^2\leq-\sqrt{\rho\epsilon}/3$, as long as $\lambda_{\min}(\mathcal{H}(\vect{x}_t))\leq-\sqrt{\rho\epsilon}$. The error probability of this assumption is provided later.

Under this assumption, \algo{QuantumSimulation} can be called for at most $\frac{81(f(\vect{x_{0}})-f^{*})}{2}\sqrt{\frac{\rho}{\epsilon^3}}\leq \frac{T}{3}$ times, for otherwise the function value decrease will be greater than $f(\vect{x_{0}})-f^{*}$, which is not possible. Then, the error probability that some calls to \algo{QuantumSimulation} fails to indicate a negative curvature is upper bounded by
\begin{align}
\frac{81(f(\vect{x_{0}})-f^{*})}{2}\sqrt{\frac{\rho}{\epsilon^3}}\cdot\delta_0=\delta.
\end{align}

Excluding those iterations that QuantumSimulation is applied, we still have at least $2T/3$ steps left, which are all accelerated gradient descent steps.

Since from $\epsilon\leq\ell^{2}/\rho$ we have $\mathscr{T}'\geq\mathscr{T}$, then we can found at least $\frac{T}{3\mathscr{T}}$ disjoint time periods, each of time interval $\mathscr{T}$. From \lem{AGD-large-gradient}, during these time intervals the Hamiltonian will decrease in total at least:
\begin{align}
\frac{T}{3\mathscr{T}}\times\mathscr{E}=f(\vect{x_{0}})-f^{*},
\end{align}
which is impossible due to \lem{Hamiltonian-decrease}, the Hamiltonian decreases monotonically for every step where quantum simulation is not called, and the overall decrease cannot be greater than $f(\vect{x_{0}})-f^{*}$.

Note that the iteration numbers $T$ satisfies:
\begin{align}
T=3\max\left \{ \frac{(f(\vect{x_{0}})-f^{*})}{\mathscr{F'}},\frac{(f(\vect{x_{0}})-f^{*})\mathscr{T}}{\mathscr{E}} \right \}=\tilde{O}\Big(\frac{(f(\vect{x}_{0})-f^{*})}{\epsilon^{1.75}}\cdot \log n\Big).
\end{align}
As for the number of queries, it can be viewed as the sum of two parts, the number of queries needed for accelerated gradient descent, denoted by $T_{1}$, and the number of queries needed for quantum simulation, denoted by $T_{2}$. Then with probability at least $1-\delta$,
\begin{align}
T_{1}=T=\tilde{O}\Big(\frac{(f(\vect{x}_{0})-f^{*})}{\epsilon^{1.75}}\cdot\log n\Big).
\end{align}
For $T_{2}$, with probability at least $1-\delta$ quantum simulation is called for at most $\frac{(f(\vect{x_{0}})-f^{*})}{\mathscr{F'}}$ times, and by \lem{simulation} it takes $\tilde{O}\big(\mathscr{T}'\log n \log^{2}(\mathscr{T}'^{2}/\epsilon)\big)$ queries to carry out each simulation. Therefore,
\begin{align}
T_{2}=\frac{3(f(\vect{x_{0}})-f^{*})}{\mathscr{F}'}\cdot\tilde{O}\big(\mathscr{T}'\log n \log^{2}(\mathscr{T}'^{2}/\epsilon)\big)=\tilde{O}\Big(\frac{(f(\vect{x}_{0})-f^{*})}{\epsilon^{1.75}}\cdot\log^2 n\Big).
\end{align}
As a result, the total query complexity $T_{1}+T_{2}$ is
\begin{align}
\tilde{O}\Big(\frac{(f(\vect{x}_{0})-f^{*})}{\epsilon^{1.75}}\cdot\log^{2}n\Big).
\end{align}
\end{proof}

\begin{remark}
Although the theorem above only guarantees that one of the iterations is an $\epsilon$-approximate second-order stationary point, it can be easily accessed by adding a proper termination condition: once the quantum simulation is called, we keep track of the point $\tilde{\vect{x}}$ prior to quantum simulation, and compare the function value at $\tilde{\vect{x}}$ with that of $\vect{x}_t$ after the perturbation. If the function value decreases by at least $\mathscr{F}'$, then the algorithm has made progress, otherwise with probability at least $1-\delta$, $\tilde{\vect{x}}$ is an $\epsilon$-approximate second-order stationary point. Doing so will add an extra register for saving the point but does not increase the asymptotic complexity.
\end{remark}


\section{Gradient Descent by the Quantum Evaluation Oracle}\label{sec:gradient-Jordan}
Another important contribution of this paper is to show how to replace the classical gradient queries by quantum evaluation queries. This is shown in the case of convex optimization~\cite{chakrabarti2020optimization,vanApeldoorn2020optimization}, and we generalize the same result to nonconvex optimization.

The idea was initiated by~\cite{jordan2005fast}. Classically, with only an evaluation oracle, the best way to construct a gradient oracle is probably to walk along each direction a little bit and compute the finite difference in each coordinate. Quantumly, a clever approach is to take the uniform superposition on a mesh around the point, query the quantum evaluation oracle (in superposition) in phase,\footnote{This can be achieved by a standard technique called phase kickback. See more details at~\cite{gilyen2019optimizing} and \cite{chakrabarti2020optimization}.} and apply the quantum Fourier transform (QFT). Due to Taylor expansion,
\begin{align}
\sum_{\vect{x}}e^{if(\vect{x})}\vect{x}\approx\sum_{\vect{x}}\bigotimes_{k=1}^{n}e^{i\frac{\partial f}{\partial x_{k}}\vect{x}_{k}}\vect{x}_{k},
\end{align}
the QFT can recover all the partial derivatives simultaneously. In this paper, we refer to Lemma 2.2 of~\cite{chakrabarti2020optimization} for a precise version of Jordan's algorithm:
\begin{restatable}{lemma}{JordanGrad}\label{lem:quantum-grad}
Let $f\colon \R^n \to \R$ be an $\ell$-smooth function specified by the evaluation oracle in \eqn{quantum-evaluation} with accuracy $\delta_{q}$, i.e., it returns a value $\tilde{f}(x)$ such that $|\tilde{f}(x)-f(x)|\leq \delta_{q}$. For any $\x\in\R^{n}$, there is a quantum algorithm that uses one query to \eqn{quantum-evaluation} and returns a vector $\tnb f(\x)$ s.t.
\begin{align}\label{eqn:quantum-grad-1}
\P\left[\|\tnb f(\x)-\nabla f(\x)\|_{2}>400\omega n \sqrt{\delta_{q} \ell}\right]<\min\Big\{\frac{n}{\omega-1},1\Big\},\qquad\forall \omega>1.
\end{align}
\end{restatable}

The main technical contribution of this section is to replace the gradient descent steps in \sec{QSimulation} by \lem{quantum-grad}. We give error bounds of gradient computation steps in \sec{error-gradient}, and give the proof details of escaping from saddle points in \sec{GD-gradient}.

\subsection{Error Bounds of Gradient Computation Steps}\label{sec:error-gradient}
We first give the following bound on gradient descent using \lem{quantum-grad}:
\begin{lemma}\label{lem:descent-lemma-19}
Let $f\colon \R^n \to \R$ be an $\ell$-smooth, $\rho$-Hessian Lipschitz function, and let $\eta\leq 1/\ell$. Then the gradient outputted by \lem{quantum-grad} satisfies that for any fixed constant $c$, with probability at least $1-\frac{n}{\frac{1}{A_{q}}\sqrt{\frac{2c}{\eta}}-1}$, any specific step of the gradient descent sequence $\{\vect{x}_{t}: \vect{x}_{t+1}\leftarrow\vect{x}_{t}-\eta\tnb \vect{x}_{t}\}$ satisfies:
\begin{align}
f(\vect{x}_{t+1})-f(\vect{x}_{t})\leq -\eta \|\nabla f(\vect{x}_{t})\|^{2}/2+c,
\end{align}
where $A_{q}=400 n\sqrt{\delta_{q} \ell}$ in the formula stands for a constant of the accuracy of the quantum algorithm.
\end{lemma}
Ideally speaking, $A_{q}$ can be arbitrarily small given a quantum computer that is accurate enough using more qubits for the precision $\delta_{q}$.
\begin{proof}
Considering our condition of $f$ being $\ell$-smooth, we have
\begin{align}
f(\vect{x}_{t+1})\leq f(\vect{x}_{t})+\nabla f(\vect{x}_{t})\cdot(\vect{x}_{t+1}-\vect{x}_{t})+ \frac{\ell}{2}\|\vect{x}_{t+1}-\vect{x}_{t}\|^{2}.
\end{align}
we use $\vect{g}(\vect{x})$ to denote the outcome of the quantum algorithm. Then by the definition of gradient descent, $\vect{x}_{t+1}-\vect{x}_{t}=\eta \vect{g}(\vect{x}_{t})$. Let $\delta[\vect{g}(\vect{x})]:=\vect{g}(\vect{x})-\nabla f(\vect{x})$. Then we have
\begin{align}
f(\vect{x}_{t+1}) &\leq f(\vect{x}_{t})+\nabla f(\vect{x}_{t})\cdot(\vect{x}_{t+1}-\vect{x}_{t})+ \frac{\ell}{2}\|\vect{x}_{t+1}-\vect{x}_{t}\|^{2} \\
&\leq f(\vect{x}_{t})-\eta \nabla f(\vect{x}_{t})\cdot(\nabla f(\vect{x}_{t})+\delta[\vect{g}(\vect{x}_{t})])+ \frac{\eta}{2}\|\nabla f(\vect{x}_{t})+\delta[\vect{g}(\vect{x}_{t})]\|^{2} \\
&=f(\vect{x}_{t})-\frac{\eta}{2} \|\nabla f(\vect{x}_{t})\|^{2}+\frac{\eta}{2}\|\delta[\vect{g}(\vect{x}_{t})]\|^{2}.
\end{align}
By \lem{quantum-grad}, for a fixed constant $c$, the value of $\frac{\eta}{2}\|\delta[\vect{g}(\vect{x}_{t})]\|^{2}$ is smaller than $c$ with probability at least $1-\frac{n}{\frac{1}{A_{q}}\sqrt{\frac{2c}{\eta}}-1}$, completing the proof.
\end{proof}

Now, we replace all the gradient queries in \algo{PGD+QS} by quantum evaluation queries, which results in \algo{PGD-Jordan}. We aim to show that if it starts at $\vect{x}_{0}$ and the value of the objective function does not decrease too much over iterations, then its whole iteration sequence $\left \{ \vect{x}_{\tau} \right \}_{\tau=0}^{t}$ will be located in a small neighborhood of $\vect{x}_{0}$. Intuitively, this is a robust version of the ``improve or localize'' phenomenon presented in~\cite{jin2019stochastic}.

 \begin{algorithm}[htbp]
\caption{Perturbed Gradient Descent with Quantum Simulation and Gradient Computation.}
\label{algo:PGD-Jordan}
\For{$t=0,1,\ldots,T$}{
Apply \lem{quantum-grad} to compute an estimate $\tnb f(\x)$ of $\nabla f(\x)$\;
\If{$\|\tnb f(\vect{x}_{t})\|\leq \epsilon$}{
$\xi\sim$QuantumSimulation($\vect{x}_{t},r_{0},\mathscr{T}'$)\;
$\Delta_t\leftarrow\frac{2\xi}{3\|\xi\|}\sqrt{\frac{\rho}{\epsilon}}$\;
$\vect{x}_{t}\leftarrow\mathop{\arg\min}_{\zeta\in\left\{\vect{x}_t+\Delta_t,\vect{x}_t-\Delta_t\right\}}f(\zeta)$\;
}
$\vect{x}_{t+1}\leftarrow \vect{x}_{t}-\eta\tnb f(\vect{x}_{t})$\;
}
\end{algorithm}

\begin{lemma}\label{lem:improve-or-localize-21}
Under the setting of \lem{descent-lemma-19}, for arbitrary $t>\tau>0$ and arbitrary constant $c$, with probability at least $1-\frac{nt}{\frac{1}{A_{q}}\sqrt{\frac{2c}{\eta}}-1}$ we have
\begin{align}
\|\vect{x}_{\tau}-\vect{x}_{0}\|\leq 2\sqrt{\eta t|f(\vect{x}_{0})-f(\vect{x}_{t})|}+2\eta t\sqrt{c},
\end{align}
if quantum simulation is not called during $[0,t]$.
\end{lemma}

\begin{proof}
Observe that
\begin{align}
\|\vect{x}_{\tau}-\vect{x}_{0}\|\leq \sum_{\tau=1}^{t} \|\vect{x}_{\tau}-\vect{x}_{\tau-1}\|.
\end{align}
Using the Cauchy-Schwartz inequality, the formula above can be converted to:
\begin{align}
\|\vect{x}_{\tau}-\vect{x}_{0}\|\leq \sum_{\tau=1}^{t}\|\vect{x}_{\tau}-\vect{x}_{\tau-1}\|\leq \Big[t\sum_{\tau=1}^{t}\|\vect{x}_{\tau}-\vect{x}_{\tau-1}\|^{2}\Big]^{\frac{1}{2}},
\end{align}
in which
\begin{align}
\vect{x}_{\tau}-\vect{x}_{\tau-1}=\eta \vect{g}(\vect{x}_{\tau-1})=\eta\nabla f(\vect{x}_{\tau-1})+\eta \delta[\vect{g}(\vect{x}_{\tau-1})],
\end{align}
which results in
\begin{align}
\|\vect{x}_{\tau}-\vect{x}_{\tau-1}\|^{2}
&\leq \eta^{2}\|\nabla f(\vect{x}_{\tau-1})\|^{2}+2\eta^{2}\nabla f(\vect{x}_{\tau-1})\cdot \delta[\vect{g}(\vect{x}_{\tau-1})]+\eta^{2}\|\delta[\vect{g}(\vect{x}_{\tau-1})]\|^{2}
\\
&\leq 2\eta^{2}\|\nabla f(\vect{x}_{\tau-1})\|^{2}+2\eta^{2}\|\delta[\vect{g}(\vect{x}_{\tau-1})]\|^{2}.
\end{align}
Go back to the first inequality,
\begin{align}
\|\vect{x}_{\tau}-\vect{x}_{0}\|\leq \Big[t\sum_{\tau=1}^{t}\|\vect{x}_{\tau}-\vect{x}_{\tau-1}\|^{2}\Big]^{\frac{1}{2}}\leq\Big[2\eta^{2} t\sum_{\tau=1}^{t}(\|\nabla f(\vect{x}_{\tau-1})\|^{2}+\|\delta[\vect{g}(\vect{x}_{\tau-1})]\|^{2})\Big]^{\frac{1}{2}}.
\end{align}
Suppose during each step from $1$ to $t$, the value of $\|\delta[\vect{g}(\vect{x}_{\tau-1})]\|^{2}$ is smaller than the fixed constant $c$. From \lem{quantum-grad}, this condition can be satisfied with probability at least $1-\frac{nt}{\frac{1}{A_{q}}\sqrt{\frac{2c}{\eta}}-1}$. Then,
\begin{align}
\|\vect{x}_{\tau}-\vect{x}_{0}\|
&\leq\Big[2\eta^{2} t\sum_{\tau=1}^{t}\Big(\|\nabla f(\vect{x}_{\tau-1})\|^{2}+\|\delta[\vect{g}(\vect{x}_{\tau-1})]\|^{2}\Big)\Big]^{\frac{1}{2}} \\
&\leq\Big[2\eta^{2} t\Big(\frac{2f(\vect{x}_{0})-2f(\vect{x}_{t})}{\eta}+2t\|\delta[\vect{g}(\vect{x}_{\tau-1})]\|^{2}\Big)\Big]^{\frac{1}{2}} \\
&\leq[4\eta t(f(\vect{x}_{0})-f(\vect{x}_{t})+\eta tc)]^{\frac{1}{2}} \\
&\leq2\sqrt{\eta t|f(\vect{x}_{0})-f(\vect{x}_{t})|}+2\eta t\sqrt{c}.
\end{align}
\end{proof}

\subsection{Escaping from Saddle Points with Quantum Simulation and Gradient Computation}\label{sec:GD-gradient}
In this subsection, we prove the result below for escaping from saddle points with both quantum simulation and gradient computation. Compared to \thm{PGD+QS-Complexity}, it reduces classical gradient queries to the same number of quantum evaluation queries.

\begin{theorem}\label{thm:ESCGDJordan}
Let $f\colon \R^n \to \R$ be an $\ell$-smooth, $\rho$-Hessian Lipschitz function. Suppose that we have the quantum evaluation oracle $U_{f}$ in \eqn{quantum-evaluation} with accuracy $\delta_{q}\leq O\Big(\frac{\delta^2\epsilon^2}{\ell n^4}\Big)$. Then \algo{PGD-Jordan} finds an $\epsilon$-approximate local minimum satisfying \eqn{eps-approx-local-min}, using
\begin{align*}
\tilde{O}\Big(\frac{(f(\vect{x}_{0})-f^{*})}{\epsilon^{2}}\cdot\log^{2}n\Big)
\end{align*}
queries to $U_{f}$ with probability at least $1-\delta$, under the following parameter choices:
\begin{align}
     \mathscr{T}'   &:= \frac{8}{(\rho\epsilon)^{1/4}}\log\Big(\frac{\ell}{\delta_0\sqrt{\rho\epsilon}}(n+2\log(3/\delta_0))\Big)     & \mathscr{F}' &:= \frac{2}{81}\sqrt{\frac{\epsilon^3}{\rho}}\\
      r_0  &:= \frac{4C_r^3}{9\mathscr{T}'^4}\Big(\frac{\delta_0}{3}\cdot\frac{1}{n^{3/2}+2C_0n\ell(\log \mathscr{T}')^{\alpha}}\Big)^2
      & \eta  &:= \frac{1}{\ell}
\end{align}
where $C_0$ and $C_r$ are constants specified in \prop{QS-effectiveness}, $\x_{0}$ is the start point, and $f^{*}$ is the global minimum of $f$.
\end{theorem}

Note that \thm{ESCGDJordan} essentially shows that the perturbed gradient descent method still converges with the same asymptotic bound if there is a small error in gradient queries. This robustness of escaping from saddle points may be of independent interest.

\begin{proof}
Set $\delta_0=\frac{1}{81(f(\vect{x}_0)-f^{*})}\sqrt{\frac{\epsilon^3}{\rho}}$ and set the quantum accuracy $\delta_q\leq\frac{1}{2\ell}\Big(\frac{\delta\epsilon}{1000n^2}\Big)^2$. Let total iteration steps $T$ to be:
\begin{align}
T=4\max\left \{ \frac{(f(\vect{x_{0}})-f^{*})}{\mathscr{F'}},\frac{2(f(\vect{x_{0}})-f^{*})}{\eta \epsilon^{2}} \right \}
=\tilde{O}\Big(\frac{(f(\vect{x}_{0})-f^{*})}{\epsilon^{2}}\cdot\log n\Big),
\end{align}
similar to the classical GD algorithm. The same to \prop{QS-effectiveness}, we set the radius $M$ of the simulation range to be $r_0/C_r$. First assume that for each $\vect{x}_t$ we apply QuantumSimulation (\algo{QuantumSimulation}),
we can successfully obtain an output $\xi$ with $\xi^{T}\mathcal{H}\xi/\|\xi\|^2\leq-\sqrt{\rho\epsilon}/3$, as long as $\lambda_{\min}(\mathcal{H}(\vect{x}_t))\leq-\sqrt{\rho\epsilon}$. The error probability of this assumption is provided later.

Under this assumption, \algo{QuantumSimulation} can be called for at most $\frac{81(f(\vect{x_{0}})-f^{*})}{2}\sqrt{\frac{\rho}{\epsilon^3}}\leq \frac{T}{4}$ times, for otherwise the function value decrease will be greater than $f(\vect{x_{0}})-f^{*}$, which is not possible. Then, the error probability that some calls to \algo{QuantumSimulation} fails to indicate a negative curvature is upper bounded by
\begin{align}
\frac{81(f(\vect{x_{0}})-f^{*})}{2}\sqrt{\frac{\rho}{\epsilon^3}}\cdot\delta_0=\delta/2.
\end{align}

Excluding those iterations that QuantumSimulation is applied, we still have $T/2$ steps left. They are either large gradient steps, $\|\nabla f(\vect{x}_{t})\|\geq \epsilon$, or $\epsilon$-approximate second-order stationary points. Within them, for each large gradient steps, by \lem{descent-lemma-19}, with probability at least
\begin{align}
1-\frac{n}{\frac{1}{400n}\sqrt{\frac{2}{\delta_{q}}\cdot\frac{\eta \epsilon^{2}}{4}}-1}=1-\frac{n}{\frac{\epsilon}{400n}\sqrt{\frac{1}{2\delta_{q}\ell}}-1}\leq 1-\delta/2,
\end{align}
the function value decrease is greater than $\eta\epsilon^{2}/4$, there can be at most $T/4$ steps with large gradients---otherwise the function value decrease will be greater than $f(\vect{x}_{0})-f^{*}$, which is impossible.

In summary, by the union bound we can deduce that with probability at least $1-\delta$, there are at most $T/2$ steps within $\mathscr{T}'$ steps after calling quantum simulation, and at most $T/4$ steps have a gradient greater than $\epsilon$. As a result, the rest $T/4$ steps must all be $\epsilon$-approximate second-order stationary points.

The number of queries can be viewed as the sum of two parts, the number of queries needed for gradient descent, denoted by $T_{1}$, and the number of queries needed for quantum simulation, denoted by $T_{2}$. Then with probability at least $1-\delta$,
\begin{align}
T_{1}=T=\tilde{O}\Big(\frac{(f(\vect{x}_{0})-f^{*})}{\epsilon^{2}}\cdot\log n\Big).
\end{align}
As for $T_{2}$, with probability at least $1-\delta$ quantum simulation is called for at most $\frac{4(f(\vect{x_{0}})-f^{*})}{\mathscr{F'}}$ times, and by \lem{simulation} it takes $\tilde{O}\big(\mathscr{T}'\log n \log^{2}(\mathscr{T}'^{2}/\epsilon)\big)$ queries to carry out each simulation. Therefore,
\begin{align}
T_{2}=\frac{4(f(\vect{x_{0}})-f^{*})}{\mathscr{F'}}\cdot\tilde{O}\big(\mathscr{T}'\log n \log^{2}(\mathscr{T}'^{2}/\epsilon)\big)=\tilde{O}\Big(\frac{(f(\vect{x}_{0})-f^{*})}{\epsilon}\cdot\log^{2}n\Big).
\end{align}
As a result, the total query complexity $T_{1}+T_{2}$ is
\begin{align}
\tilde{O}\Big(\frac{(f(\vect{x}_{0})-f^{*})}{\epsilon^{2}}\cdot\log^{2}n\Big).
\end{align}
\end{proof}

\thm{QuantumSimulationAGD} and \thm{ESCGDJordan} together imply the main result \thm{main-intro} of this paper.
\begin{remark}
One may notice that in \sec{gradient-Jordan}, we only demonstrated the robustness of \algo{PGD+QS} where the classical gradient oracle is replaced by the quantum evaluation oracle. We argue that the same argument holds for \algo{PAGDQS} because the difference between \algo{PGD+QS} and \algo{PAGDQS} only exists in large gradient steps, while the relative error caused by Jordan's algorithm is small since the absolute error remains to be a constant. Hence, in principle \algo{PAGDQS} satisfies the similar robustness property compared to \algo{PGD+QS} under the change of gradient oracles.
\end{remark}


\section{Numerical Experiments}\label{sec:numerical}
In this section, we provide numerical results that demonstrate the power of quantum simulation for escaping from saddle points. Due to the limitation of current quantum computers, we simulate all quantum algorithms numerically on a classical computer (with Dual-Core Intel Core i5 Processor, 8GB memory). Nevertheless, our numerical results strongly assert the quantum speedup in small to intermediate scales. All the numerical results and plots are obtained by MATLAB 2019a.

In the first two experiments, we look at the wave packet evolution on both quadratic and non-quadratic potential fields. Before bringing out numerical results and related discussions, we want to briefly discuss the leapfrog scheme~\cite{gray1996symplectic}, which is the technique we employed for numerical integration of the Schr\"odinger equation. We discretize the Schr\"odinger equation as a linear system of an ordinary differential equation (for details, see \sec{simulation}):
\begin{equation} \label{eqn:discrete_Sch}
    i \frac{\d \Psi}{\d t} = H \Psi,
\end{equation}
where $\Psi\colon [0,T] \to \mathbb{C}^N$ is a vector-valued function in time. We may have a decomposition $\Psi(t) = Q(t) + iP(t)$ for $Q,P\colon [0,T] \to \R^N$ being the real and imaginary part of $\Psi$, respectively. Then plugging the decomposition into the ODE \eqn{discrete_Sch}, we have a separable $N$-body Hamiltonian system
\begin{equation}
    \begin{cases}
    \dot{Q} = HP;\\
    \dot{P} = -HQ.
    \end{cases}
\end{equation}
The optimal integration scheme for solving this Hamiltonian system is the symplectic integrator~\cite{gray1996symplectic}, and we use a second-order leapfrog integrator for separable canonical Hamiltonian systems~\cite{mauger2020leapfrog} in this section. In all of our PDE simulations, we fix the spatial domain to be $\Omega = \{(x,y):|x|\le 3, |y|\le 3\}$ and the mesh number to be $512$ on each edge.

\subsection{Dispersion of the Wave Packet}
In \prop{QSimulation}, we showed that a centered Gaussian wave packet will disperse along the negative curvature direction of the saddle point. In the numerical simulation presented in \fig{qsimulation_quadratic}, we have a potential function $f_{1}(x,y) = -x^2/2+3y^2/2$ and the initial wave function as described in \prop{QSimulation} ($r = 0.5$). In each subplot, the Gaussian wave packet (i.e., modulus square of the wave function $\Phi(t,x)$) at a specific time is shown. The quantum evolution ``squeezes'' the wave packet along the $x$-axis: the variance of the marginal distribution on the $x$-axis is 0.25, 0.33, 0.68 at time $t = 0, 0.5, 1$, respectively.

\begin{figure}[ht]
    \centering
    \includegraphics[width=1\textwidth]{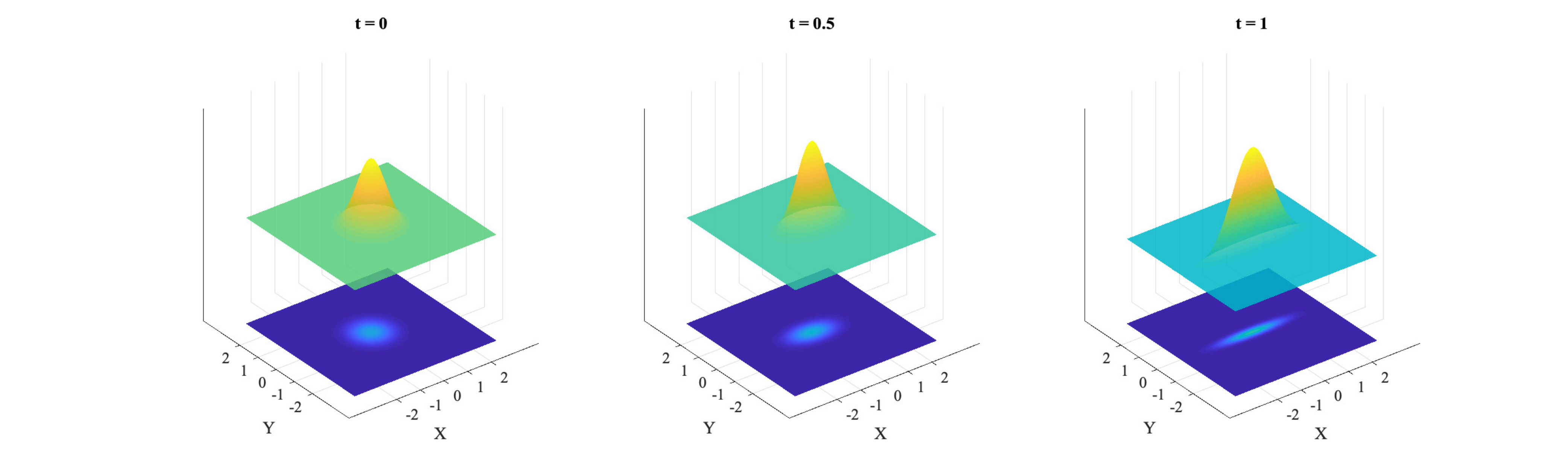}
    \caption{Dispersion of wave packet over the potential field $f_{1}(x,y)$. We use the finite difference method (5-point stencil) and the Leapfrog integration to simulate the Schr\"odinger equation \eqn{Schrodinger-appendix} on a square domain (center $= (0,0)$, edge $=6$), up to $T=1$. The mesh number is $512$ on each edge. The average runtime for this simulation is 43.7 seconds.}
    \label{fig:qsimulation_quadratic}
\end{figure}

In the preceding experiment, we have provided a numerical simulation of the dispersion of the Gaussian wave packet on a quadratic potential field. Next, we only require that the function is Hessian-Lipschitz near the saddle point. This is enough to promise that the second-order Taylor series is a good approximation near a small neighborhood of the saddle point.

\subsection{Quantum Simulation on Non-quadratic Potential Fields}
Now, we explore the behavior of the wave packet on non-quadratic potential fields. It is worth noting that: (1) the wave packet is not necessarily Gaussian during the time evolution; (2) for practical reason, we will truncate the unbounded spatial domain $\R^2$ to be a bounded square $\Omega$ and assume Dirichlet boundary conditions ($\Phi(t,x) = 0$ on $\partial \Omega$ for all $t\in [0,T]$). Nevertheless, it is still observed that the wave packet will be mainly confined to the ``valley'' on the landscape which corresponds to the direction of the negative curvature.

We will run quantum simulation (\algo{QuantumSimulation}) near the saddle point of two non-quadratic potential landscapes. The first one is $f(x,y) = \frac{1}{12}x^4 - \frac{1}{2}x^2+\frac{1}{2}y^2$. The Hessian matrix of $f(x,y)$ is
\begin{equation}
    \nabla^{2} f(x,y) = \begin{pmatrix}
    x^2  - 1 & 0 \\
    0 & 1\\
    \end{pmatrix}.
\end{equation}
It has a saddle point at $(0,0)$ and two global minima $(\pm \sqrt{3},0)$. The minimal function value is $-3/4$. This is the landscape used in the next experiment in which a comparison study between quantum and classical is conducted. We claimed that the wave packet will remain (almost) Gaussian at $t_e = 1.5$. This claim is confirmed by the numerical result illustrated in \fig{qsimulation_nonquadratic}. The wave packet has been ``squeezed'' along the $x$-axis, the negative curvature direction. Compared to the uniform distribution in a ball used in PGD, this ``squeezed'' bivariant Gaussian distribution assigns more probability mass along the $x$-axis, thus allowing escaping from the saddle point more efficiently.

\begin{figure}[ht]
    \centering
    \includegraphics[width=1\textwidth]{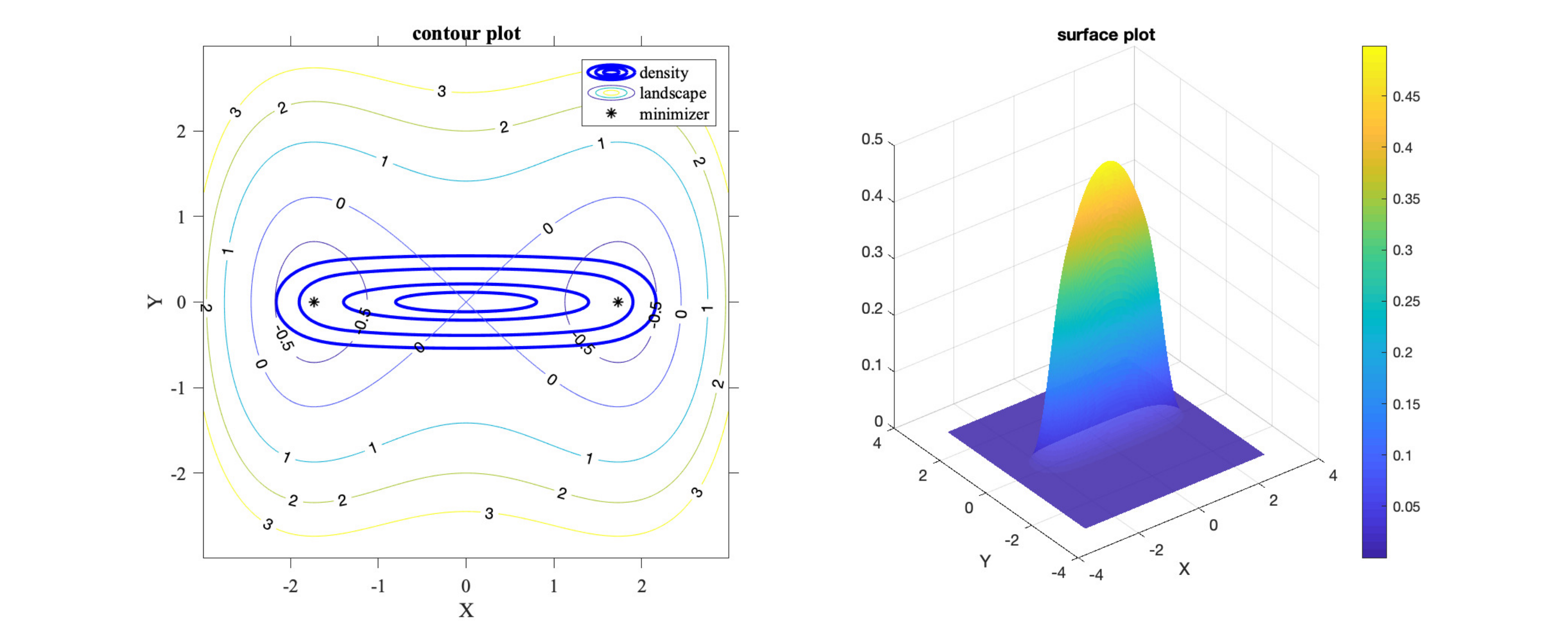}
    \caption{Quantum simulation on landscape 1: $f(x,y) = \frac{1}{12}x^4 - \frac{1}{2}x^2+\frac{1}{2}y^2$. Parameters: $r_0 = 0.5$, $t_e = 1.5$. Left: The contour of the landscape is placed on the background with labels being function values; the thick blue contours illustrate the wave packet at $t_e = 1.5$ (i.e., modulus square of the wave function $\Phi(t_e,x,y)$).\\Right: A surface plot of the same wave packet at $t_e = 1.5$. The average runtime for this simulation is 60.70 seconds.}
    \label{fig:qsimulation_nonquadratic}
\end{figure}

The second landscape we explore is $g(x,y) = x^3-y^3-2xy+6$. Its Hessian matrix is
\begin{equation}
    \nabla^{2} g(x,y) = \begin{pmatrix}
    6x & -2 \\
    -2 & -6y\\
    \end{pmatrix}.
\end{equation}
It has a saddle point at $(0,0)$ with no global minimum. This objective function has a circular ``valley'' along the negative curvature direction $(1,1)$, and a ``ridge'' along the positive curvature direction $(1,-1)$. We aim to study the long-term evolution of the Gaussian wave packet on the landscape restricted on a square region. The evolution of the wave packet is illustrated in \fig{qsimulation_nonquadratic2}. In a small time scale (e.g., $t = 1$), the wave packet disperses down the valley on the landscape, and it preserves a bell shape; waves are reflected from the boundary and an interference pattern can be observed near the upper and left edges of the square. Dispersion and interference coexist in the plot at $t = 2$, in which the wave packet splits into two symmetric components, each locates in a lowland. Since the total energy is conserved in the quantum-mechanical system, the wave packet bounces back at $t = 5$, but is blurred due to wave interference. In the whole evolution in $t\in [0,5]$, the wave packet is confined to the valley area of the landscape (even after bouncing back from the boundary). This evidence suggests that Gaussian wave packet is able to adapt to more complicated saddle point geometries.

\begin{figure}[ht]
    \centering
    \includegraphics[width=0.8\textwidth]{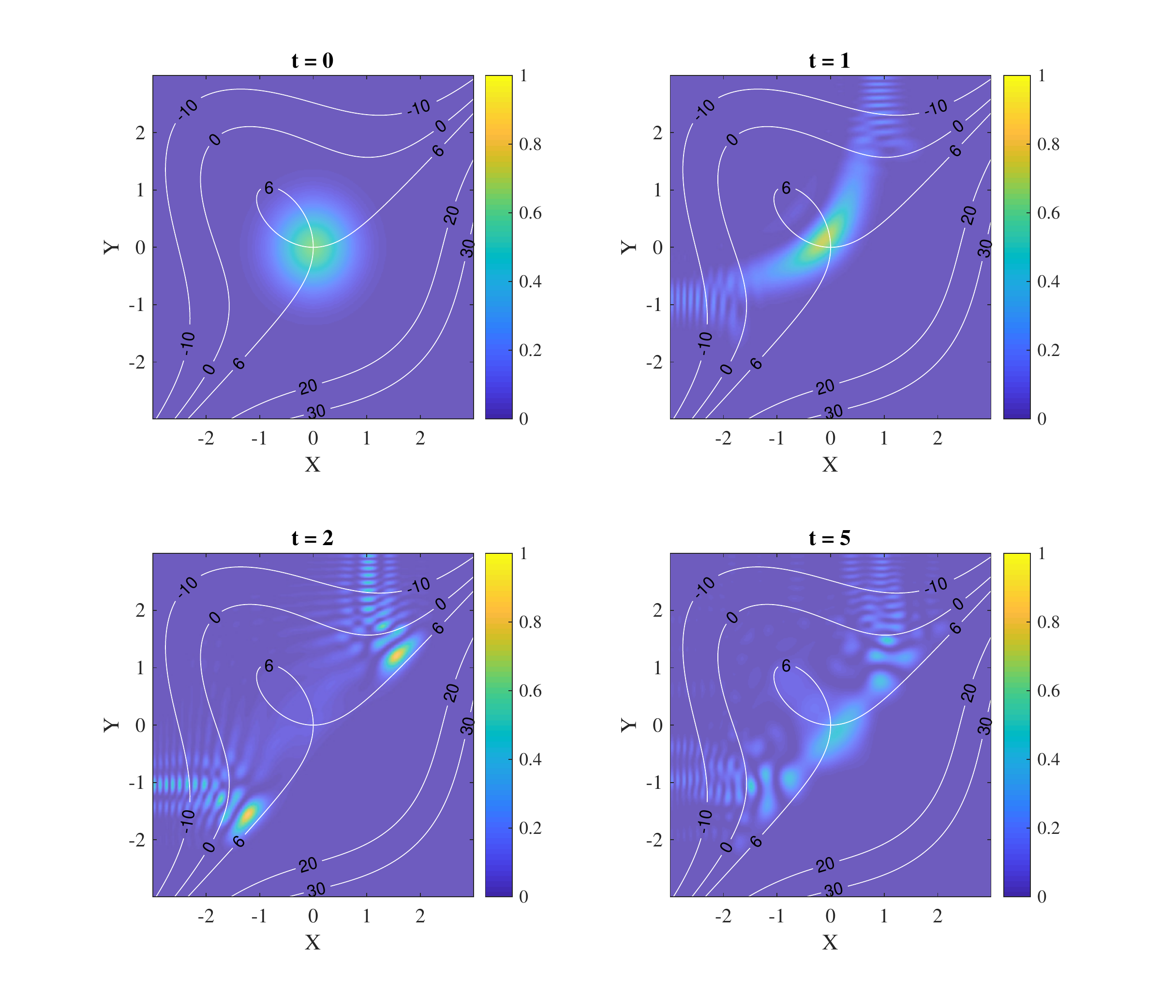}
    \caption{Quantum simulation on landscape 2: $g(x,y) = x^3-y^3-2xy+6$. Parameters: $r_0 = 0.5$, $t_e = 5$. In each subplot, a colored contour plot of the wave packet at a specific time is shown, and the landscape contour is placed on top of the wave packet  for quick reference. The average runtime for this simulation is 209.95 seconds.}
    \label{fig:qsimulation_nonquadratic2}
\end{figure}

\subsection{Comparison Between PGD and \algo{PGD+QS}}
In addition to the numerical study of the evolution of wave packets, we compare the performance of the PGD algorithm~\cite{jin2017escape} with \algo{PGD+QS} on a test function $f_{2}(x,y) = \frac{1}{12}x^4 - \frac{1}{2}x^2+\frac{1}{2}y^2$.

In this experiment and the last one in this section, we only implement a mini-batch from the whole algorithm (for both classical PGD and PGD with quantum simulation). In fact, a mini-batch is good enough for us to demonstrate the power of quantum simulation as well as the dimension dependence in both algorithms. A \textit{mini-batch} in the experiment is defined as follows:

\begin{itemize}
    \item Classical algorithm (PGD) mini-batch~\cite[following Algorithm 4 of][]{jin2019stochastic}: $x_0$ is uniformly sampled from the ball $B_0(r)$ (saddle point at the origin), and then run $\mathscr{T}_c$ gradient descent steps to obtain $x_{\mathscr{T}_c}$. Record the function value $f(x_{\mathscr{T}_c})$. Repeat this process $M$ times. The resulting function values are presented in a histogram.

    \item Quantum algorithm mini-batch (following \algo{PGD+QS}): Run the quantum simulation with evolution time $t_e$ to generate a multivariate Gaussian distribution centered at $0$. $x_0$ is sampled from this multivariate Gaussian distribution. Run $\mathscr{T}_q$ gradient descent steps and record the function value $f(x_{\mathscr{T}_q})$. Repeat this process $M$ times. The resulting function values are also presented in a histogram, superposed to the results given by the classical algorithm.
\end{itemize}

The experimental results from 1000 samples are illustrated in \fig{result_nonquadratic}. Although the test function is non-quadratic, the quantum speedup is apparent.

\begin{figure}[ht]
    \centering
    \includegraphics[width=0.8\textwidth]{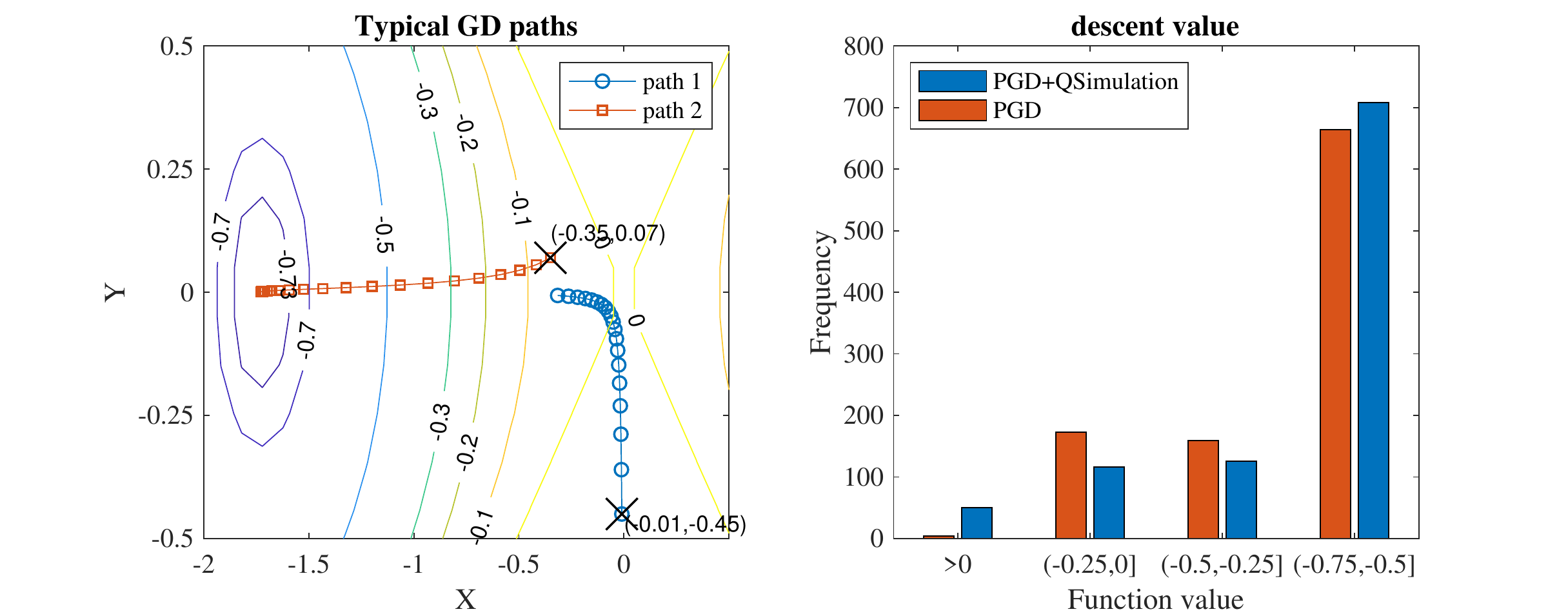}
    \caption{Left: Two typical gradient descent paths on the landscape of $f_{2}$ illustrated as a contour plot. Path 1 (resp.~2) starts from $(-0.01,0.45)$ (resp.~$(-0.35,0.07)$); both have step length $\eta = 0.2$ and $T = 20$ iterations. Note that path 2 approaches the local minimum $(-\sqrt{3},0)$, while path 1 is still far away. In PGD, path 1 and 2 will be sampled with equal probability by the uniform perturbation, whereas in \algo{PGD+QS}, the dispersion of the wave packet along the $x$-axis
    enables a much higher probability of sampling a path like path 2 (that approaches the local minimum).\\
    Right: A histogram of function values $f_{2}(x_{\mathscr{T}_c})$ (PGD) and $f_{2}(x_{\mathscr{T}_q})$ (\algo{PGD+QS}). We set step length $\eta =0.05$, $r = 0.5$ (ball radius in PGD and $r_0$ in \algo{QuantumSimulation}), $M = 1000$, $\mathscr{T}_c = 50$, $\mathscr{T}_q = 10$, $t_e = 1.5$. Although we run five more times of iterations in PGD, there are still over $70\%$ of gradient descent paths arriving the neighborhood of the local minimum, while there are less than $70\%$ paths in \algo{PGD+QS} approaching the local minimum. The average runtime of this experiment is 0.02 seconds.}
    \label{fig:result_nonquadratic}
\end{figure}

\subsection{Dimension Dependence}
Recall that $n$ is the dimension of the problem. Classically, it has been shown in~\cite{jin2019stochastic} that the PGD algorithm requires $O(\log^4 n)$ iterations to escape from saddle points; however, quantum simulation for time $O(\log n)$ suffices in our \algo{PGD+QS} by \thm{PGD+QS-Complexity}. The following experiment is designed to compare this dimension dependence of PGD and \algo{PGD+QS}. We choose a test function $h(x) = \frac{1}{2}x^T \mathcal{H} x$ where $\mathcal{H}$ is an $n$-by-$n$ diagonal matrix: $\mathcal{H} = \diag(-\epsilon, 1, 1, ..., 1)$. The function $h(x)$ has a saddle point at the origin, and only one negative curvature direction. Throughout the experiment, we set $\epsilon = 0.01$. Other hyperparameters are: dimension $n \in \N$, radius of perturbation $r>0$, classical number of iterations $\mathscr{T}_c$, quantum number of iterations $\mathscr{T}_q$, quantum evolution time $t_e$, number of samples $M \in \N$, and GD step size (learning rate) $\eta$. For the sake of comparison, the iteration numbers $\mathscr{T}_c$ and $\mathscr{T}_q$ are chosen in a manner such that the statistics of the classical and quantum algorithms in each category of the histogram in \fig{type1} are of similar magnitude.

\begin{figure}[ht]
    \centering
    \includegraphics[width=0.8\textwidth]{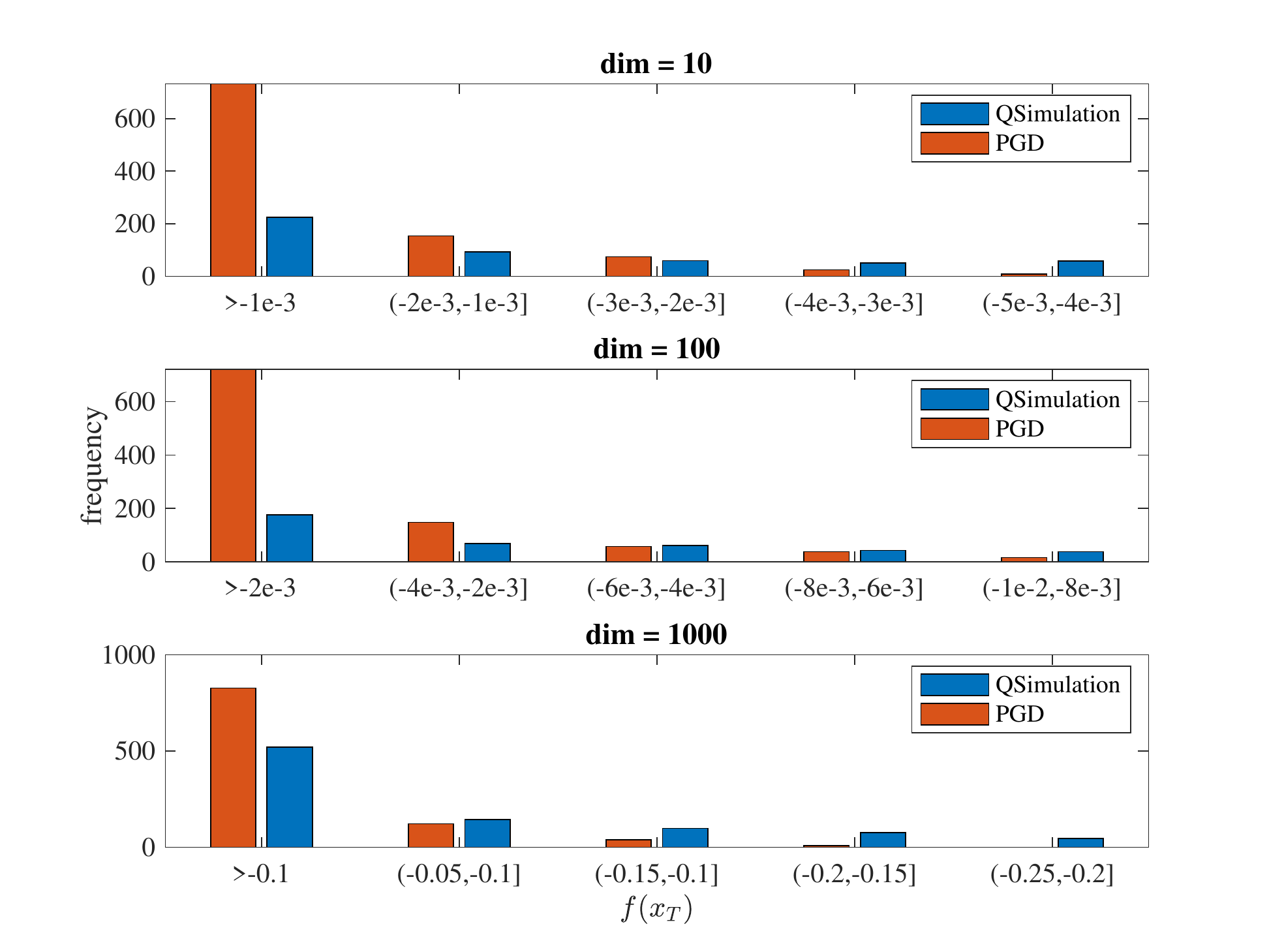}
    \caption{Dimension dependence of classical and quantum algorithms. We set $\epsilon = 0.01$, $r = 0.1$, $n = 10^p$ for $p = 1, 2, 3$. Quantum evolution time $t_e = p$, classical iteration number $\mathscr{T}_c = 50p^2 + 50$, quantum iteration number $\mathscr{T}_q = 30p$, and sample size $M = 1000$. The average runtime for this simulation is 90.92 seconds.}
    \label{fig:type1}
\end{figure}

The numerical results are illustrated in \fig{type1}. The number of dimensions varies drastically from $10$ to $1000$, while the distribution patterns in all three subplots are similar: setting $\mathscr{T}_c = \Theta(\log^2 n)$ and $\mathscr{T}_q = \Theta(\log n)$, the PGD with quantum simulation outperforms the classical PGD in the sense that more samples can escape from the saddle point (as they have lower function values). At the same time, under this choice of parameters, the performance of the classical PGD is still comparable to that of the PGD with quantum simulation, i.e., the statistics in each category are of similar magnitude. This numerical evidence might suggest that for a generic problem, the classical PGD method in \cite{jin2019stochastic} has better dimension dependence than $O(\log^4 n)$.


\section*{Acknowledgement}
We thank Andrew M. Childs, Andr{\'a}s Gily{\'e}n, Aram W. Harrow, Jin-Peng Liu, Ronald de Wolf, and Xiaodi Wu for helpful discussions. We also thank anonymous reviewers for helpful suggestions on earlier versions of this paper. JL was supported by the National Science Foundation (grant CCF-1816695). TL was supported by an IBM PhD Fellowship, an QISE-NET Triplet Award (NSF grant DMR-1747426), the U.S. Department of Energy, Office of Science, Office of Advanced Scientific Computing Research, Quantum Algorithms Teams program, NSF grant PHY-1818914, and a Samsung Advanced Institute of Technology Global Research Partnership.

\providecommand{\bysame}{\leavevmode\hbox to3em{\hrulefill}\thinspace}


\appendix
\section{Auxiliary Lemmas}\label{append:auxiliary-lemmas}
In this appendix, we collect all auxiliary lemmas that we use in the proofs.

\subsection{Schr\"{o}dinger Equation with a Quadratic Potential}\label{append:QSimulation-append}
In this subsection, we prove several results that lay the foundation of the quantum algorithm described in \sec{QSimulation}.

\standQSimulation*

\begin{proof}
 Due to the well-posedness of the Schr\"odinger equation, if we find a solution to the initial value problem \eqref{eqn:Schrodinger}, this solution is unique. We take the following ansatz
    \begin{equation} \label{eqn:ansatz}
        \Phi(t,x) = \left(\frac{1}{\pi}\right)^{1/4} \frac{1}{\sqrt{\delta(t)}}\exp(-i\theta(t)) \exp\left(\frac{-x^2}{2\delta(t)^2}\right),
    \end{equation}
    with $\theta(0) = 0$, $\delta(0) = \sqrt{2}$.

    In this Ansatz, the probability density $p_\lambda(t,x)$, i.e., the modulus square of the wave function, is given by
    \begin{equation} \label{eqn:density}
        p_\lambda(t,x):= |\Phi(t,x)|^2 = \frac{1}{\sqrt{\pi}} \frac{1}{|\delta(t)|} \exp\Big(2 \operatorname{Im}(\theta(t))\Big) \exp\Big(-x^2 \operatorname{Re}(1/ y(t))\Big),
    \end{equation}
    where $y(t) = \delta^2(t)$.

    If the ansatz \eqref{eqn:ansatz} solves the Schr\"odinger equation, we will have the conservation of probability, i.e., $\|\Phi(t,x)\|^2 = 1$ for all $t \geq 0$; in other words, the $\int_\R p_\lambda(t,x) \d x = 1$ for all $t \geq 0$. It is now clear that \eqref{eqn:density} is the density of a Gaussian random variable with zero mean and variance
    \begin{equation} \label{eqn:variance0}
        \sigma^2(t;\lambda) = \frac{1}{2\operatorname{Re}(1/ y(t))}.
    \end{equation}
    Therefore, it is sufficient to compute $y(t)$ in order to obtain the distribution of the quantum particle at time $t\geq 0$. For simplicity, we will not compute the global phase $\theta(t)$ as it is not important in the the variance.

    Substituting the ansatz \eqref{eqn:ansatz} to \eqref{eqn:Schrodinger} with potential function $f(x) = \frac{\lambda}{2}x^2$, and introducing change of variables $y(t) = \delta^2(t)$, we attain the following system of ordinary differential equations

    \begin{equation} \label{eqn:ode}
    \begin{cases}
    y' +i\lambda y^2 - i = 0,\\
    \theta' = \frac{i}{4}\frac{y'}{y} + \frac{1}{2}\frac{1}{y},\\
    \theta(0) = 0, y(0) = 2.
    \end{cases}
    \end{equation}

    \noindent\textbf{Case 1: $\lambda = 0$.} The system \eqref{eqn:ode} is linear with solutions
    \begin{equation}
    y(t) = 2+it.
    \end{equation}

    It follows that
    \begin{equation}
    \frac{1}{y(t)} = \frac{2}{4+t^2} - i \frac{t}{4+t^2},
    \end{equation}
    And by Equation \eqref{eqn:variance0}, the variance is
    \begin{equation}
    \sigma^2(t;0) = 1+\frac{t^2}{4}.
    \end{equation}

    \noindent\textbf{Case 2: $\lambda \neq 0$.} The equation $y' + i\lambda y^2 - i = 0$ in \eqref{eqn:ode} is a Riccati equation. Using the standard change of variable $y = \frac{-i}{\lambda} \frac{u'}{u}$, we transfer the Riccati equation into a second-order linear equation
    \begin{equation} \label{eqn:ode2}
    u'' + \lambda u = 0.
    \end{equation}
    Clearly, the sign of $\lambda$ matters.

    \noindent\textbf{Case 2.1: $\lambda > 0$.} Let $\alpha = \sqrt{\lambda}$, the solution to \eqref{eqn:ode2} is $u(t) = c_1 e^{i\alpha t} + c_2 e^{-i\alpha t}$ ($c_1, c_2$ are constants), and
    \begin{equation}
        y(t) = \frac{-i}{\lambda} \frac{u'}{u} = \frac{1}{\alpha}\frac{c_1 e^{i\alpha t} - c_2 e^{-i\alpha t} }{c_1 e^{i\alpha t} + c_2 e^{-i\alpha t}}.
    \end{equation}
    Provided the initial condition $y(0) = 2$, we choose $c_1 = 1$, $\beta := c_2 = (1-2\alpha)/(1+2\alpha)$, and it turns out that
    \begin{equation} \label{eqn:case2_1}
        y(t) = \frac{1}{\alpha}\Big(\frac{e^{2i\alpha t} - \beta }{e^{2i\alpha t} + \beta }\Big).
    \end{equation}
    By \eqref{eqn:variance0} and \eqref{eqn:case2_1}, we attain the variance when $\lambda > 0$.

    \noindent\textbf{Case 2.2: $\lambda < 0$.} Let $\alpha = \sqrt{-\lambda} > 0$, similar as Case 2.1, we have
    \begin{equation} \label{eqn:case2_2}
        y(t) = \frac{i}{\alpha} \frac{e^{2\alpha t} - \beta}{e^{2\alpha t} + \beta},
    \end{equation}
    where $\beta = \frac{1+2i\alpha}{1-2i\alpha}$. And the variance $\sigma(t;\lambda)$ for $\lambda<0$ can be obtained from \eqref{eqn:variance0} and \eqref{eqn:case2_2}.
    \end{proof}
    
\begin{remark}
Essentially, the three cases $\lambda=0$, $\lambda>0$, and $\lambda<0$ in Eq.~\eqn{variance_standard} can be written as a simple expression following \eqn{case2_1} and \eqn{case2_2}. Here we present these cases separately to explicitly demonstrate that when $\lambda<0$, the variance $\sigma^2(t;\lambda)$ grows exponentially fast in $t$.
\end{remark}

Furthermore, we prove that the argument applies to $n$-dimensional cases in general:
\begin{lemma}[$n$-dimensional evolution] \label{lem:standard_multidim_QSimulation}
Let $\mathcal{H}$ be an $n$-by-$n$ symmetric matrix with diagonalization $\mathcal{H} = U^T \Lambda U$, with $\Lambda = \diag(\lambda_1, ..., \lambda_n)$ and $U$ an orthogonal matrix. Suppose a quantum particle is in an $n$-dimensional potential field $f(\vect{x}) = \frac{1}{2}\vect{x}^T \mathcal{H}\vect{x}$ with initial state $\Phi(0,x)=(\frac{1}{2\pi})^{n/4} \exp(-\|\vect{x}\|^2/4)$; in other words, the initial position of this quantum particle follows multivariate Gaussian distribution $\mathcal{N}(0,I)$. Then, at any time $t \ge 0$, the position of the quantum particle still follows multivariate Gaussian distribution $\mathcal{N}(0, \Sigma(t))$, with the covariance matrix
\begin{equation} \label{eqn:variance_n}
    \Sigma(t) = U^T \diag(\sigma^2(t;\lambda_1), ..., \sigma^2(t;\lambda_n))U.
\end{equation}
The function $\sigma(t;\lambda)$ is defined in \eqref{eqn:variance_standard}.
\end{lemma}

\begin{proof}
The proof follows the same idea in \lem{1_dim_standard_QSimulation}. We take the following ansatz
    \begin{equation} \label{eqn:ansatz_ndim}
        \Phi(t,\vect{x}) = \left(\frac{1}{\pi}\right)^{n/4} \left(\det D(t)\right)^{-1/4}\exp(-i\theta(t)) \exp\left[-\frac{1}{2}\vect{x}^T \big(D(t)\big)^{-1} \vect{x}\right],
    \end{equation}
    with $\theta(0) = 0$, $D(0) = \sqrt{2}I$, and $D(t) = U^T \diag(\delta^2_1(t),...,\delta^2_n(t))U$.

The global phase parameter $\theta(t)$, together with the factor $\left(\frac{1}{\pi}\right)^{n/4} \left(\det D(t)\right)^{-1/4}$, will contribute to a scalar factor in the probability density function such that the $L^2$-norm of the wave function \eqref{eqn:ansatz_ndim} will remain unit $1$. It is the matrix $D(t)$ that controls the covariance matrix (see Eqn.~\ref{eqn:variance1}). Regarding this, we do not delve into the derivation of $\theta(t)$ in this proof.

Substituting the ansatz \eqref{eqn:ansatz_ndim} to the Schr\"odinger equation \eqref{eqn:Schrodinger}, we have the following system of ordinary differential equations:

\begin{equation} \label{eqn:delta_ode}
    \frac{\d }{\d t}\left(D(t)^{-1}\right) + i D(t)^{-2} - i\mathcal{H} = 0,
\end{equation}
\begin{equation}
    \dot{\theta} = \frac{i}{4}\left(\det D(t)\right)^{-1} \frac{\d }{\d t}\left(D(t)\right)+\frac{1}{2}\Tr[D(t)^{-1}].
\end{equation}
We immediately observe that Eq. \eqref{eqn:delta_ode} is a decoupled system
\begin{equation}
    \frac{\d}{\d t}\left(\frac{1}{\delta_j(t)^2}\right) + i \frac{1}{(\delta_j(t))^4} - i \lambda_j = 0, \text{  for }j = 1,..., n.
\end{equation}
Again, introduce change of variables $y_j(t) = \delta^2_j(t)$, we have
\begin{equation} \label{eqn:decouple2}
    \dot{y}_j + i \lambda_j y^2-i=0, \text{  for }j = 1,..., n.
\end{equation}
They are precisely the same as the first equation in \eqref{eqn:ode}, thus the calculation of one-dimensional case in \lem{1_dim_standard_QSimulation} applies directly to \eqref{eqn:decouple2}.

Given the ansatz \eqref{eqn:ansatz_ndim}, it is clear that the probability density of the quantum particle in $\R^n$ is an $n$-dimensional Gaussian with mean $0$ and covariance matrix
\begin{equation} \label{eqn:variance1}
    \Sigma(t) = \left( 2\operatorname{Re} D^{-1}(t)\right)^{-1} = U^T \left(\frac{1}{2\operatorname{Re}(1/y_1(t))},...,\frac{1}{2\operatorname{Re}(1/y_n(t))}\right) U.
\end{equation}
It follows from \eqref{eqn:variance0} and \eqref{eqn:variance_standard} that the covariance matrix is given as \eqref{eqn:variance_n}.
\end{proof}

Finally, we state the following proposition with different scales:
\begin{restatable}{proposition}{QSimulation}\label{prop:QSimulation}
Let $\mathcal{H}$ be an $n$-by-$n$ symmetric matrix with diagonalization $\mathcal{H} = U^T \Lambda U$, with $\Lambda = \diag(\lambda_1, ..., \lambda_n)$ and $U$ an orthogonal matrix. Suppose a quantum particle is in an $n$-dimensional potential field $f(\vect{x}) = \frac{1}{2}\vect{x}^T \mathcal{H}\vect{x}$ with the initial state being
\begin{equation}
\Phi(0,\vect{x})=\Big(\frac{1}{2\pi}\Big)^{n/4} r^{-n/2}\exp(-\|\vect{x}\|^2/4r^2);
\end{equation}
in other words, the initial position of the particle follows multivariate Gaussian distribution $\mathcal{N}(0,r^2 I)$. The time evolution of this particle is governed by \eqref{eqn:Schrodinger-appendix}. Then, at any time $t \ge 0$, the position of the quantum particle still follows multivariate Gaussian distribution $\mathcal{N}(0, r^2\Sigma(t))$, with the covariance matrix
\begin{equation}
    \Sigma(t) = U^T \diag(\sigma^2(t;\lambda_1), ..., \sigma^2(t;\lambda_n))U.
\end{equation}
The function $\sigma(t;\lambda)$ is the same as in \eqref{eqn:variance_standard}.
\end{restatable}

\begin{proof}
Here, we only prove the one-dimensional case, as the $n$-dimensional case follows almost the same manner, together with a similar argument from the proof of \lem{standard_multidim_QSimulation}.
Let $\Phi(t,x)$ be the wave function as in \lem{1_dim_standard_QSimulation}, namely, it satisfies the standard Schr\"odinger equation \eqref{eqn:Schrodinger}.
Define $\Psi(t,x) = \frac{1}{\sqrt{r}}\Phi(t,\frac{x}{r})$. Since $\|\Phi(t,\cdot)\|^2 = 1$ for all $ t \ge 0$, the factor $\frac{1}{\sqrt{r}}$ ensures the $L^2$-norm of $\Psi(t,x)$ is always 1.

We claim that $\Psi(t,x)$ satisfies the modified Schr\"odinger equation \eqref{eqn:Schrodinger-appendix}. To do so, we substitute $\Psi(t,x)$ back to \eqref{eqn:Schrodinger-appendix}. Its LHS is just $i \frac{\partial}{\partial_t} \frac{1}{\sqrt{r}}\Phi(t,x/r)$, whereas the RHS is
\begin{align}
    \Big[-\frac{r^2}{2}\Delta+\frac{1}{r^2}f(\vect{x})\Big]\Psi(t,x) = \frac{1}{\sqrt{r}}\Big[-\frac{1}{2}\Delta+\frac{1}{2}\left(\frac{\vect{x}}{r}\right)^T \mathcal{H} \left(\frac{\vect{x}}{r}\right) \Big] \Phi\left(t,\frac{x}{r}\right).
\end{align}
Since $\Phi(t,x)$ satisfies \eqref{eqn:Schrodinger}, it turns out that the LHS equals to the RHS. Furthermore, the variance of $\Phi(t,x)$ is $\sigma^2(t;\lambda)$, and that of $\Psi(t,x) =\frac{1}{\sqrt{r}} \Phi(t,x/r)$ is simply $r^2 \sigma^2(t;\lambda)$.
\end{proof}

Throughout the discussion, we only concern the evolution of the wave packet when it happens to center on the saddle point. However, in reality, the exact location of the saddle point is rarely known and the initial Gaussian wave may be slightly off the saddle point. In the following proposition, we investigate this more general situation in which the potential function is shifted by a distance of $d$. It turns out that the wave packet remains Gaussian with exactly the same rate of dispersion in its variance, while the mean of the Gaussian wave behaves like the trajectory of a classical particle, i.e., governed by the Hamiltonian mechanics $\ddot{X} = - \nabla f(X)$. Thus, we believe the source of quantum speedup in our algorithm is the variance dispersion along the negative curvature direction.

\begin{proposition}\label{prop:off_center}
	Suppose a quantum particle is in a one-dimensional potential field $f(x)=\frac{\lambda}{2}(x-d)^{2}$ with initial state $\Phi(0,x)=(\frac{1}{2\pi})^{1/4}\exp(-x^{2}/4)$; in other words, the initial position of this quantum particle follows the standard normal distribution $\mathcal{N}(0,1)$. The time evolution of this particle is governed by \eqref{eqn:Schrodinger}. Then, at any time $t \ge 0$, the position of the quantum particle still follows normal distribution $\mathcal{N}\left(\mu(t;\lambda),\sigma^2(t;\lambda)\right)$, where the mean $\mu(t;\lambda)$ is given by
	\begin{equation} \label{eqn:mean_standard}
		\mu(t;\lambda) = \begin{cases}
			0 & (\lambda = 0),\\
			d(1 - \cos(\alpha t)) & (\lambda > 0, \alpha = \sqrt{\lambda}),\\
			d(1 - \cosh(\alpha t)) & (\lambda < 0, \alpha = \sqrt{-\lambda}),\\
		\end{cases}
	\end{equation}
	while the variance $\sigma^2(t;\lambda)$ is exactly the same as in \eqref{eqn:variance_standard}.
\end{proposition}
\begin{proof}
	The main idea of the proof is to use the undetermined coefficient method similar to the proof of \lem{1_dim_standard_QSimulation}, though we will use a different ansatz with more parameters:
	\begin{equation}\label{eqn:off_cent_ansatz}
		\Phi(t,x) = \exp\left(-a(t) x^2 + b(t)x + c(t)\right),
	\end{equation}
	where $a(t)$, $b(t)$, and $c(t)$ are complex-valued functions. For simplicity, the normalization constant is absorbed in the $c(t)$ term. The probability density $p_\lambda(t,x)$, i.e., the modulus square of the wave function, is then given by
	\begin{equation}
		p_\lambda (t,x) := |\Phi(t,x)|^2 = \exp\left(-\frac{\big(x-\mathscr{B}(t)/\mathscr{A}(t)\big)^2}{1/2\mathscr{A}(t)} + \big(\mathscr{B}(t)^2/2\mathscr{A}(t) + 2\mathscr{C}(t)\big)\right),
	\end{equation}
	where $\mathscr{A}(t)$, $\mathscr{B}(t)$, and $\mathscr{C}(t)$ are the real parts of the functions $a(t)$, $b(t)$, and $c(t)$, respectively. One can readily observe that $p_\lambda (t,x)$ is a Gaussian density function with mean and variance being
	\begin{equation}\label{eqn:off_cent_m_v}
		\begin{cases}
			\mu(t;\lambda) = \frac{\mathscr{B}(t)}{2\mathscr{A}(t)},\\
			\sigma^2(t;\lambda) = \frac{1}{4\mathscr{A}(t)}.
		\end{cases}
	\end{equation}
	It turns out that the distribution of the quantum particle is completely determined by the mean $\mu(t)$ and variance $\sigma^2(t)$ if we can show that the ansatz function \eqref{eqn:off_cent_ansatz} indeed solves the Schr\"odinger equation \eqref{eqn:Schrodinger} with a potential field $f(x) = \frac{\lambda}{2}(x-d)^2$.
	
Substituting the ansatz \eqref{eqn:off_cent_ansatz} to the Schr\"odinger equation \eqref{eqn:Schrodinger}, we obtain the following system of ordinary differential equations:
	\begin{equation}\label{eq:off_cent_system}
		\begin{cases}
			-i \dot{a} = -2a^2 + \frac{\lambda}{2},\\
			i\dot{b} = 2ab - \lambda d,\\
			i\dot{c} = a - \frac{1}{2}b^2 + \frac{\lambda}{2}d^2,
		\end{cases}
	\end{equation}
	subject to the initial condition $a(0) = 1/4$, $b(0) = 0$, and $c(0) = -\log(2\pi)/4$.
	The last equation says $c(t)$ can be directly integrated as long as $a(t)$ and $b(t)$ are known. In other words, $c(t)$ exists given that $a(t)$ and $b(t)$ are determined, and we do not care about the exact value of $c(t)$ because it sheds no light on either the mean $\mu(t;\lambda)$ nor the variance $\sigma^2(t;\lambda)$. To prove the lemma, it suffices to calculate $a(t)$ and $b(t)$.
	
	The first equation in the system \eqref{eq:off_cent_system} is a Riccati equation; by the change of variable $a = -\frac{i}{2} \frac{\dot{u}}{u}$, the Riccati equation is transformed into a second-order linear equation $\ddot{u} + \lambda u = 0$. Then, similarly, we shall discuss three cases $\lambda = 0$, $\lambda > 0$, and $\lambda < 0$. Here, we only do the $\lambda  > 0$ case, as the other two cases are solved following essentially the same procedures.
	
	Before we proceed with the calculation of $a(t)$, we discuss how the change of variable $a = -\frac{i}{2} \frac{\dot{u}}{u}$ simplifies the second equation in the system \eqref{eq:off_cent_system}.  With the change of variable into $i\dot{b} = 2ab - \lambda d$ and proper algebraic manipulation, we end up with the nice form
	\begin{align}
		\dot{u}b + u\dot{b} = i\lambda d u,
	\end{align}
	Note that the left hand side is simply $\frac{\d}{\d t}(ub)$, and hence the function $b(t)$ can be expressed in terms of $u(t)$:
	\begin{equation}\label{eqn:eq_b}
		b(t) = i \lambda d\cdot\frac{\int^t_0 u(s) \d s + C}{u(t)},
	\end{equation}
	where $C$ is a constant.
	
	Now, we are ready to compute both the mean and variance for the case $\lambda > 0$. Suppose $\alpha = \sqrt{\lambda}$, we have
	\begin{equation}\label{eqn:eq_u}
		u(t) = e^{i\alpha t} + c e^{-i\alpha t}, \text{ with } c = (1-2\alpha)/(1+2\alpha).
	\end{equation}
	This particular choice of $c$ will give rise to the function $a(t)$ satisfying the initial condition $a(0) = 1/4$, which reads
	\begin{align}
		a(t) = \frac{\alpha}{2} \frac{e^{2i\alpha t} - c}{e^{2i\alpha t} + c}, \text{ with }  c = (1-2\alpha)/(1+2\alpha).
	\end{align}
	
	Similarly, we substitute the solution of $u(t)$ \eqref{eqn:eq_u} back into the formula for $b(t)$ \eqref{eqn:eq_b}, together with the initial condition $b(0) = 0$, we can write down the closed form of $b(t)$:
	\begin{align}
		b(t) = \alpha r \frac{e^{2i\alpha t} - c + (c-1)e^{i\alpha t}}{e^{2i\alpha t} + c}, \text{ with }  c = (1-2\alpha)/(1+2\alpha).
	\end{align}
	
	 The real parts of $a(t)$ and $b(t)$ can then be computed as follows
	 \begin{equation}
	 	\begin{cases}
	 		\mathscr{A}(t) = \Re\left(a(t)\right) = \frac{(1-c^2)\alpha}{2\left(1 + c^2 + 2\cos(2\alpha t)\right)},\\
	 		\mathscr{B}(t) =  \Re\left(b(t)\right) =  \alpha d \frac{(1-c^2)\big(1-\cos(\alpha t)\big)}{1 + c^2 + 2\cos(2\alpha t)},
	 	\end{cases}
	 \end{equation}
	and the mean $\mu(t;\lambda)$ and variance $\sigma^2(t;\lambda)$ follows from \eqref{eqn:off_cent_m_v}.
\end{proof}

\subsection{Bounding the deviation from perfect Gaussian in quantum evolution}\label{append:evol_deviation}
In what follows, we will use $\|\cdot\|_p$ to denote the $L^p$-norm of an integrable function $g\colon \Omega \to \R$:
\begin{align}
	\|g\|_p := \left(\int_{\Omega} |g|^p ~\d x \right)^{1/p},
\end{align}
where $1 \le p < \infty$. For a continuous function $g\colon \Omega \to \R$, the $L^\infty$ norm is $\|g\|_\infty = \sup_{x\in \Omega} |g(x)|$. For a finite-dimensional vector $\vec{v}$, we simply use $\|\vec{v}\|$ to denote its $\ell^2$-norm (or the Euclidean norm):
\begin{align}
	\|\vec{v}\| := \left(\sum_j |v_j|^2 \right)^{1/2}.
\end{align}
For a vector-valued function $G\colon \Omega \to \R^n$, we also define its $L^p$-norm for $1 \le p < \infty$:
\begin{align}
	\|G\|_p := \left(\int_\Omega \sum^n_{j=1}  |G_j(x)|^p~ \d x\right)^{1/p},
\end{align}
where $G_j(x)$ is the $j$-th component of the function $G(x)$. The $L^\infty$-norm is defined in the same manner: $\|G\|_\infty = \max_{1\le j \le n} \|G_j\|_\infty$.

First, we prove the following vector norm error bound of quantum simulation:
\begin{lemma}[Vector norm error bound] \label{lem:vector-norm-bound}
	Let $H_1$, $H_2$ be two Hermitian operators and $H = H_1+H_2$. Then, for any $t>0$ and an arbitrary vector $\ket{\varphi}$, we have
	\begin{align}\label{eqn:vector-norm-bound}
		\left\|e^{-iH_1 t}e^{-iH_2 t}\ket{\varphi} - e^{-iH t}\ket{\varphi}\right\| \le \frac{t^2}{2}\sup_{\tau_1, \tau_2 \in [0,t]} \left\|[H_1, H_2] e^{-iH_2 \tau_2} e^{-iH_1 \tau_1}  \ket{\varphi}\right\|.
	\end{align}
\end{lemma}
\begin{proof}
	By \cite[Proposition 15]{childs2019theory}, we have the variation-of-parameter formula
	\begin{align}
	\hspace{-2mm} e^{-iH_1 t}e^{-iH_2 t} =  e^{-iH t} + \int^t_0 \d \tau_1 \int^{\tau_1}_0 \d \tau_2~ e^{-iH(t-\tau_1)} e^{-iH_2 \tau_1} e^{-iH_2 \tau_2} [H_1, H_2] e^{-iH_2 \tau_2} e^{-iH_1 \tau_1}.
	\end{align}
Thus, for an arbitrary vector $\ket{\varphi}$, we have
	\begin{align}
		&\left(e^{-iH_1 t}e^{-iH_2 t}-  e^{-iH t}\right)\ket{\varphi}\nonumber\\
&\qquad=  \int^t_0 \d \tau_1 \int^{\tau_1}_0 \d \tau_2~ e^{-iH(t-\tau_1)} e^{-iH_2 \tau_1} e^{-iH_2 \tau_2} [H_1, H_2] e^{-iH_2 \tau_2} e^{-iH_1 \tau_1} \ket{\varphi}.
	\end{align}
	Since the spectral norm of the vector in the integrand is upper bounded by
	\begin{align}
		\sup_{\tau_1, \tau_2 \in [0,t]} \left\|[H_1, H_2] e^{-iH_2 \tau_2} e^{-iH_1 \tau_1}  \ket{\varphi}\right\|,
	\end{align}
	 and $\int^t_0 \d \tau_1 \int^{\tau_1}_0 \d \tau_2 = \frac{t^2}{2}$, we obtain the desired vector norm error bound \eqn{vector-norm-bound}.
\end{proof}

Second, we observe the following fact:
\begin{theorem}[{\cite[Theorem 2, informal]{bourgain1999growth}}] \label{thm:sobolev-norm-growth}
	For Schr\"odinger equations of the form
	\begin{align}
		i\frac{\partial}{\partial t} u + \Delta u + V(x,t)u = 0,
	\end{align}
	defined over an arbitrary finite-dimensional space with periodic boundary condition, let $u(x,t)$ be the solution at time $t$. If $V(x,t)$ is smooth in space and periodic in time, and the initial condition $u(x,0)$ is smooth, then we have
	\begin{align}
	\|\nabla u(t)\|_2 \le C (\log t)^\alpha  \|\nabla u(0)\|_2, 	
	\end{align}
	where $C$ and $\alpha$ are absolute constants.
\end{theorem}
\begin{remark}
	The original Theorem 2 in \cite{bourgain1999growth} actually proved the logarithmic growth in Sobolev norm $\|u(t)\|_{H^s}$ for all $s>0$,  while we only cite the special case $s=1$. The $\|\nabla u(0)\|_2$ term was absorbed in the constant factor in the original statement, while we feel necessary to expand it out because it may introduce dependence on $n$ and $r_0$. It is worth noting that the theorem was proven for two-dimensional Schr\"odinger equations with quasi-periodic potential field $V(x,t)$, while it has been made clear in the context that this result holds for arbitrary-dimensional cases if $V$ is periodic. Bourgain also explicitly discussed the periodic-$V$ case in \cite{bourgain1999growth2}.
\end{remark}

\begin{corollary}\label{cor:growth}
	For a quadratic function of the form $f_q = \frac{1}{2} (\vect{x} - \tilde{\vect{x}})^T \mathcal{H}(\vect{x} - \tilde{\vect{x}}) + F$ where $\mathcal{H}$ is a Hermitian matrix and $F$ is a constant, consider the Schr\"odinger equation of the form
	\begin{align}
		i\frac{\partial}{\partial t} \Phi = \left[-\frac{r^2_0}{2} \Delta  + \frac{1}{r^2_0}f_q \right]\Phi,
	\end{align}
	with periodic boundary conditions and initial condition $\Phi_0(x)$ defined in \eqref{eqn:ground_state_Phi0} (i.e., the initial state of the quantum simulation \algo{QuantumSimulation}), then we have
	\begin{align}
		\|\nabla \Phi(t)\|_2 \le C\sqrt{\frac{n}{r_0}}(\log t)^\alpha,
	\end{align}
	where $C$ and $\alpha$ are absolute constants.
\end{corollary}
\begin{proof}
	Note that the constant $F$ just adds a global phase to the solution which does not influence either $\|\Phi(t)\|_2$ or $\|\nabla \Phi(t)\|$, and the Schr\"odinger equation is translation-invariant under $\vect{x} \to \vect{x}- \tilde{\vect{x}}$, we may assume without loss of generality that $f_q = \frac{1}{2}\vect{x}^T \mathcal{H}\vect{x}$.
	
	Define a new function $u(\vect{x},t) = \Phi \left(\frac{r_0\vect{x}}{\sqrt{2}}, t\right)$, and it is straightforward to verify that
	\begin{align}
			i\frac{\partial}{\partial t} u + \Delta u - \frac{1}{r^2_0}f_q\left(\frac{r_0 \vect{x}}{\sqrt{2}}\right)u = 0.
	\end{align}
	Note that the function $f_q(\vect{x})$ is quadratic, so $\frac{1}{r^2_0}f_q\left(\frac{r_0 \vect{x}}{\sqrt{2}}\right) = \frac{1}{2}f_q(\vect{x})$, which is a constant multiple of $f_q$. Thus, we may directly invoke \thm{sobolev-norm-growth} to yield
	\begin{align}
		\|\nabla u(t)\|_2 \le C\left(\log t\right)^\alpha \|\nabla u(0)\|_2,
	\end{align}
	where the $\|\nabla u(0)\|_2$ can be directly calculated as follows:
	\begin{align}
		\|\nabla u(0)\|_2 &\le \left(\sum^n_{j=1} \int_{\R^n} |u_{x_j}(\vect{x},0)|^2 ~\d \vect{x}\right)^{1/2} = \frac{r_0}{\sqrt{2}}\left(\sum^n_{j=1} \int_{\R^n} |(\Phi_0)_{x_j}(r_0\vect{x}/\sqrt{2},0)|^2 ~\d \vect{x}\right)^{1/2}\\
		 &= \frac{\sqrt{r_0}}{2^{1/4}}\left(\sum^n_{j=1} \int_{\R^n} |(\Phi_0)_{x_j}(\vect{x},0)|^2 ~\d \vect{x}\right)^{1/2}\\
		&=  \frac{\sqrt{r_0}}{2^{1/4}}  \frac{1}{2r^2_0}\left(\sum^n_{j=1} \int_{\R} \frac{1}{\sqrt{2\pi} r_0} e^{-(x_j - \tilde{x}_j)^2/2r^2_0} (x_j - \tilde{x}_j)^2 ~\d x_j\right)^{1/2} = \frac{1}{2^{5/4}} \sqrt{\frac{n}{r_0}}.
	\end{align}
	Absorbing the $2^{-5/4}$ factor into the absolute constant $C$, we complete the proof.
\end{proof}

Now, we are ready to prove \lem{deviation-from-quadratic}, our result of bounding the deviation from perfect Gaussian in quantum evolution.

\DeviationFQuadratic*

\begin{proof}
	Define the following (Hermitian) operators:
	\begin{align}
		A = -\frac{r^2_0}{2}\Delta, ~~B = \frac{1}{r^2_0} f,~~ B' = \frac{1}{r^2_0} f_q,
	\end{align}
	\begin{equation}\label{eqn:def_H_E}
		H = A+B, ~H' = A+B', ~E = H - H' = \frac{1}{r^2_0}(f-f_q).
	\end{equation}
	
	Let $\ket{\Phi(t)} = e^{-iHt}\ket{\Phi_0}$ be the wave function generated by the quantum simulation with potential field $f$ and evolution time $t$, and similarly, $\ket{\Phi'(t)} := e^{-iH't}\ket{\Phi_0}$ as the wave function generated by the quantum simulation with potential field $f_{q}$ and the evolution time $t$.
	
	By \lem{vector-norm-bound}, and notice that $E$ is a scalar-valued function, we have
	\begin{align}
			\left\|e^{-iE t_e}\ket{\Phi'(t_e)} - \ket{\Phi(t_e)}\right\|_2 &\le \frac{t^2_e}{2}\sup_{\tau_1, \tau_2 \in [0,t]} \left\|[H', E] e^{-iE \tau_2} e^{-iH' \tau_1}  \ket{\Phi_0}\right\|_2 \\
			& = \frac{t^2_e }{2}\sup_{\tau_1 \in [0,t]} \left\|[H', E] e^{-iH' \tau_1}  \ket{\Phi_0}\right\|_2.
	\end{align}
	Denote $\ket{\Psi(\tau_1)} := e^{-iH' \tau_1}  \ket{\Phi_0}$. Note that $[H', E] = [A+B', E]$ and $B'$ commutes with $E$, we have
	\begin{align}
		\sup_{\tau_1 \in [0,t]} \left\|[H', E] \Psi(\tau_1)\right\|_2 &= \frac{1}{2}\sup_{\tau_1 \in [0,t]} \left\| [-\Delta, f-f_q] \Psi(\tau_1) \right\|_2 \\
		&= \frac{1}{2}\sup_{\tau_1 \in [0,t]} \left\| -\Delta(f-f_q) \Psi(\tau_1) - 2\nabla (f-f_q) \cdot \nabla \Psi(\tau_1)\right\|_2\\
		&\le \frac{1}{2}\|\Delta(f-f_q) \|_\infty +  \|\nabla(f-f_q) \|_\infty \|\nabla \Psi(\tau_1) \|_2.
	\end{align}
	The second equality follows from the fact that $[-\Delta, g]\varphi = -(\Delta g) \varphi - 2\nabla g \cdot \nabla \varphi$ for smooth functions $g$ and $\varphi$. The last step follows from the triangle inequality (and the fact that $\|\Psi(\tau_1)\| = 1$).
	By the $\rho$-Hessian Lipschitz condition of $f$, we have
	\begin{align}
|\Delta (f(\vect{x})-f_q(\vect{x}))|&= \left|\tr(\nabla^{2} f(\vect{x}) - \nabla^{2} f_q(\vect{x}))\right| = \left|\tr(\nabla^{2} f(\vect{x}) - \nabla^{2} f(\tilde{\vect{x}}))\right|\\
  &\le n \|\nabla^{2} f - \nabla^{2} f(\tilde{\vect{x}})\| \le n^{3/2}\rho M,
	\end{align}
where the second equality holds because $f_q$ is a quadratic form and $\nabla^{2}f_q(\vect{x}) = \mathcal{H} =  \nabla^{2} f(\tilde{\vect{x}})$. Note that the diameter of the hypercube domain is $n^{1/2}M$, and the last step follows from the $\rho$-Hessian Lipschitz condition. It turns out that 
\begin{align}
	\|\Delta(f-f_q)\|_{\infty} \le n^{3/2}\rho M.
\end{align}

Next, we bound the $L^\infty$-norm of the gradient of $f-f_q$:
	\begin{align}
		\|\nabla f - \nabla f_q\|_\infty &\le \sup_{\vect{x}} \|\nabla f(\vect{x}) - \nabla f_q(\vect{x})\| = \sup_{\vect{x}} \|\nabla f (\vect{x}) - \mathcal{H}(\vect{x} - \tilde{\vect{x}})\| \\
		& \le \sup_{\vect{x}} \|\nabla f (\vect{x})\| + \sup_{\vect{x}} \|\mathcal{H}(\vect{x} - \tilde{\vect{x}})\|,\label{eqn:part_0}
	\end{align}
	where the last step uses the triangle inequality. Note that $\tilde{\vect{x}}$ is a stationary point of $f$, so $\nabla f(\tilde{\vect{x}}) = 0$. By the $\ell$-smoothness condition of $f$, we obtain
	\begin{align}\label{eqn:part_a}
		 \sup_{\vect{x}} \|\nabla f (\vect{x})\| =  \sup_{\vect{x}} \|\nabla f (\vect{x}) - \nabla f(\tilde{\vect{x}}) \| \le \ell \sup_{\vect{x}}\|\vect{x} - \tilde{\vect{x}}\| \le \ell n^{1/2} M.
	\end{align}
	
	Meanwhile, the $\ell$-smoothness of $f$ implies that $\|\nabla^2 f(\vect{x})\| \le \ell$ for all $\vect{x} \in \R^n$, therefore $\|\mathcal{H}\|\le \ell$ and 
	\begin{align}\label{eqn:part_b}
		 \sup_{\vect{x}} \|\mathcal{H}(\vect{x} - \tilde{\vect{x}})\| \le \ell n^{1/2}M.
	\end{align}

	Plugging \eqn{part_a} and \eqn{part_b} to \eqn{part_0}, we end up with 
	\begin{align}
		\|\nabla (f -f_q)\|_\infty \le 2\ell n^{1/2}M.
	\end{align}
	
	The upper bound for $\sup_{\tau_1} \|\nabla \Psi(\tau_1) \|_2$ is given by \cor{growth}. Combining all three bounds, we end up with
	\begin{align}\label{eqn:trotter_bound}
			\left\|e^{-iE t_e}\ket{\Phi'(t_e)} - \ket{\Phi(t_e)}\right\|_2 \le \left( \frac{\sqrt{n}\rho}{2} + \frac{2C \ell}{\sqrt{r_0}} (\log t_e)^\alpha \right)\frac{n M t^2_e}{2}.
	\end{align}

	In what follows, we will simply write $\Psi'$ for $\Psi'(t_e)$. We also denote $\ket{\Psi''} := e^{-iEt_e} \ket{\Psi'}$. Note that $e^{-iEt_e}$ is actually a scalar function with modulus $1$, hence the two wave functions $\ket{\Psi'}$ and $\ket{\Psi''}$ yield the same probability density, i.e., $|\Psi'|^2 = |\Psi''|^2$. By the definition of total variation distance,
	\begin{align}
		\hspace{-2.5mm}TV(\mathbb{P}_{\xi}, \mathbb{P}_{\xi'}) &= 	TV(|\Psi|^2, |\Psi''|^2) \\
		&= \frac{1}{2} \int_{x\in\R^{n}} \left| \Psi \overline{\Psi} - \Psi'' \overline{\Psi''}\right| \d x\\
		& \le \frac{1}{2} \int_{x\in\R^{n}} \left| (\Psi - \Psi'')\overline{\Psi} \right| \d x + \frac{1}{2} \int_{x\in\R^{n}} \left| \Psi'' (\overline{\Psi} - \overline{\Psi''}) \right| \d x \\
		& \le \left(\int_{x\in\R^{n}} |\Psi - \Psi''|^2 \d x\right)^{1/2} \le \left( \frac{\sqrt{n}\rho}{2} + \frac{2C \ell}{\sqrt{r_0}} (\log t_e)^\alpha \right)\frac{n M t^2_e}{2}.
	\end{align}
\end{proof}
\subsection{Variance of Gaussian Wave Packets}\label{append:h-tail}
Although the variance of the Gaussian wave packet $\sigma(\lambda;t)$ is explicitly given in \eqn{variance_standard}, it is a bit heavy to use in analysis. In this subsection, we prove several lemmas that can be utilized to estimate the variance $\sigma(\lambda;t)$. Based on these lemmas, it is then possible to quantify the performance of \algo{PGD+QS}.

\begin{lemma}\label{lem:variance_estimate}
    When $\lambda > 0$,
    \begin{equation}
    \min \Big\{1, \frac{1}{2\alpha}\Big\} \le \sigma(t;\lambda) \le \max \Big\{1, \frac{1}{2\alpha}\Big\}.
    \end{equation}
    When $\lambda < 0$, let $\alpha = \sqrt{-\lambda}$,
    \begin{equation}
    \frac{1}{\sqrt{2}}\varphi(t;\alpha) \le \sigma(t;\lambda) \le \varphi(t;\alpha),
    \end{equation}
    with $\varphi(t;\alpha) = \frac{1}{2\alpha} \sinh(\alpha t) + \cosh(\alpha t)$.
\end{lemma}

\begin{proof}
    The first estimate follows from $\cos2\alpha t \in [0,1]$, while the second estimate follows from the inequality
  \begin{align}
  \frac{a+b}{2} \le \sqrt{\frac{a^2+b^2}{2}} \le \frac{a+b}{\sqrt{2}}.
  \end{align}
\end{proof}

\begin{lemma}\label{lem:h-tail}
When $\lambda<0$,
\begin{align}
\sigma^{2}(t;\lambda)\geq 1+\frac{t^{2}}{4}.
\end{align}
\end{lemma}

\begin{proof}
Recall \eqn{variance_standard} that $\sigma(t;\lambda)$ equals to:
\begin{align}
\sigma^{2}(t;\lambda)=\frac{(1-e^{2\alpha t})^2 + 4\alpha^2(1+e^{2\alpha t})^2}{16\alpha^2 e^{2\alpha t}},
\end{align}
in which $\alpha = \sqrt{-\lambda}$. The equation above can be converted to:
\begin{align}
\sigma^{2}(t;\lambda)
&=\frac{(1+4\alpha^{2})e^{4\alpha t}+(1+4\alpha^{2})-2(1-4\alpha^{2})e^{2\alpha t}}{16\alpha^{2}e^{2\alpha t}}\\
&=\frac{(1+4\alpha^{2})e^{2\alpha t}+(1+4\alpha^{2})e^{-2\alpha t}-2(1-4\alpha^{2})}{16\alpha^{2}}.
\end{align}
We denote $\mu:=2\alpha t$. Note that $\mu>0$. By the Taylor expansion of $e^{\mu}$ with Lagrange form of remainder, there exists real numbers $\zeta,\xi\in (0,\mu)$ such that
\begin{align}
e^{\mu}&=1+\mu+\frac{\mu^{2}}{2}+\frac{\mu^{3}}{6}+\frac{e^{\zeta}}{24}\mu^{4}; \\
e^{-\mu}&=1-\mu+\frac{\mu^{2}}{2}-\frac{\mu^{3}}{6}+\frac{e^{-\xi}}{24}\mu^{4}.
\end{align}
Adding these two equations, we have
\begin{align}
e^{\mu}+e^{-\mu}\geq 2+\mu^{2}+\frac{\mu^{4}}{24}(e^{\zeta}+e^{-\xi})\geq 2+\mu^{2}+\frac{\mu^{4}}{24}(1+e^{-\mu})\geq 2+\mu^{2}.
\end{align}
In other words,
\begin{align}
e^{2\alpha t}+e^{-2\alpha t}\geq 2+(2\alpha t)^{2},
\end{align}
which results in
\begin{align}
\frac{(1+4\alpha^{2})e^{2\alpha t}+(1+4\alpha^{2})e^{-2\alpha t}-2(1-4\alpha^{2})}{16\alpha^{2}}
&\geq\frac{(1+4\alpha^{2})(2+4\alpha^{2}t^{2})-2(1-4\alpha^{2})}{16\alpha^{2}}\\
&\geq\frac{16\alpha^{2}+4\alpha^{2}t^{2}}{16\alpha^{2}}\\
&=1+\frac{t^{2}}{4};
\end{align}
or equivalently,
\begin{align}
\sigma^{2}(t;\lambda)\geq 1+\frac{t^{2}}{4}.
\end{align}
\end{proof}

In the proof of \prop{QS-effectiveness}, we will also use the following fact about multivariate Gaussian distributions:
\begin{lemma}[{\cite[Proposition 1]{hsu2012tail}}]\label{lem:positive-tail}
Let $A\in \mathbb{R}^{m\times n}$ be a matrix, and let $\Sigma:=A^{T}A$. Let $\vect{x}=(x_{1},\cdots,x_{n})$ be an isotropic multivariate Gaussian random vector with mean zero. For all $t>0$:
\begin{align}
\mathbb{P}\Big(\|A\vect{x}\|^{2}>\tr(\Sigma)+2\sqrt{\tr(\Sigma^{2})t}+2\|\Sigma\|t\Big)\leq e^{-t}.
\end{align}
\end{lemma}

\subsection{Existing Lemmas}\label{append:existing-lemma}
In this subsection, we list existing lemmas from~\cite{jin2018accelerated,jin2019stochastic} that we use in our proof.

First, we use the following lemma for the large gradient scenario of gradient descent method:
\begin{lemma}[{\cite[Lemma 19]{jin2019stochastic}}]\label{lem:descent-lemma}
If $f(\cdot)$ is $\ell$-smooth and $\rho$-Hessian Lipschitz, $\eta=1/\ell$, then the gradient descent sequence $\{\vect{x}_{t}\}$ satisfies:
\begin{align}
f(\vect{x}_{t+1})-f(\vect{x}_{t})\leq \eta \|\nabla f(\vect{x})\|^{2},
\end{align}
for any step $t$ in which quantum simulation is not called.
\end{lemma}
The next lemmas are frequently used in the large gradient scenario of the accelerated gradient descent method:
\begin{lemma}[{\cite[Lemma 7]{jin2018accelerated}}]\label{lem:AGD-large-gradient}
Consider the setting of \thm{QuantumSimulationAGD}. If we have $\|\nabla f(\vect{x}_{\tau})\|\geq \epsilon$ for all $\tau\in[0,\mathscr{T}]$, then there exists a large enough positive constant $c_{A0}$, such that if we choose $c_{A}\geq c_{A0}$, by running \algo{PAGDQS} we have $E_{\mathscr{T}}-E_{0}\leq -\mathscr{E}$, in which $\mathscr{E}=\sqrt{\frac{\epsilon^{3}}{\rho}}\cdot c_{A}^{-7}$, and $E_{\tau}$ is defined as:
\begin{align}\label{eqn:Lemma-A7}
E_{\tau}:=f(\vect{x}_{\tau})+\frac{1}{2\eta'}\|\vect{v}_{\tau}\|^{2}
\end{align}
where $\eta'=\frac{1}{4\ell}$ as in \thm{QuantumSimulationAGD}.
\end{lemma}
Note that this lemma is not exactly the same as Lemma 7 of \cite{jin2018accelerated}: to be more specific, they have an extra $\iota^{-5}$ term appearing in the $\mathscr{E}$. However, this term actually only appears when we need to escape from a saddle point using the original AGD algorithm. In large gradient scenarios where the gradient is greater than $\epsilon$, it does not make a difference if we ignore this $\iota^{-5}$ term.

\begin{lemma}[{\cite[Lemma 4 and Lemma 5]{jin2018accelerated}}]\label{lem:Hamiltonian-decrease}
Assume that the function $f$ is $\ell$-smooth. Consider the setting of \thm{QuantumSimulationAGD}, for every iteration $\tau$ where quantum simulation was not called, we have
\begin{align}
E_{\tau+1}\leq E_{\tau},
\end{align}
where $E_{\tau}$ is defined in \eqn{Lemma-A7} in \lem{AGD-large-gradient}.
\end{lemma}
The correctness of these two lemmas above is guaranteed by two mechanisms. If the function does not have a large negative curvature between $\vect{x}_{t}$ and $\vect{y}_{t}$ in the current iteration, the AGD will simply make the Hamiltonian decrease efficiently. Otherwise, the Negative-Curvature-Exploitation procedure in \lin{NCE} of \algo{PAGDQS} will be triggered (same as in \cite{jin2018accelerated}) and decrease the Hamiltonian by either finding the minimum function value in the nearby region of $\vect{x}_{t}$ if $\vect{v}_{t}$ is small, or directly resetting $\vect{v}_{t}=0$ if it is large.

\end{document}